\documentclass{ptephy_v2}
\usepackage{booktabs}
\usepackage{placeins}
\usepackage{hyperref}
\hypersetup{colorlinks=true,linkcolor=blue,citecolor=blue,urlcolor=blue}
\graphicspath{{./}}

\newtheorem{theorem}{Theorem}[section]

\newcommand{\End}{\mathrm{End}}
\newcommand{\tr}{\mathrm{tr}}
\newcommand{\HS}{\mathrm{HS}}
\newcommand{\Uhat}{\hat{U}}
\newcommand{\Util}{\tilde{U}}
\newcommand{\id}{\mathbf{1}}
\newcommand{\CP}{\mathbb{CP}}
\newcommand{\VolCP}{\mathrm{Vol}}
\makeatletter
\g@addto@macro\ps@headings{%
  \def\@oddfoot{\vbox{\vskip16pt\hbox to \textwidth{}\hbox to \textwidth{\fontsize{10}{12}\selectfont\hfill\thepage\hfill}}}%
  \def\@evenfoot{\vbox{\vskip16pt\hbox to \textwidth{}\hbox to \textwidth{\fontsize{10}{12}\selectfont\hfill\thepage\hfill}}}%
  \def\@oddhead{}\def\@evenhead{}}
\renewcommand{\titlepageheadline}{}
\renewcommand{\titlepagefootline}{}
\makeatother

\begin{document}

\title{Construction and geometric properties of the fuzzy $\mathbb{CP}^3$ onion}

\author{Kohei Uno*}
\affil{Department of Physics, Hiroshima University, 1-3-1 Kagamiyama, Higashi-Hiroshima, Hiroshima 739-8526, Japan \email{koheiuno0508@gmail.com}}

\author[1]{Tomohiro Inagaki}

\begin{abstract}%
We consider the matrix model of a seven-dimensional quantum space built from
concentric fuzzy $\CP^3$'s of increasing matrix size, and show that its continuum
limit is the metric cone $C(\CP^3)$ over $\CP^3$, not flat $\mathbb{R}^7$. Every
ingredient is fixed before any comparison, so this geometry is a
prediction rather than an input. The two geometries agree in the $s$ wave ($p=0$) but
separate in the higher angular sectors (harmonics on $\CP^3$ of degree
$p\geq1$): in each sector tested, the spectrum
of the free Laplacian converges at second order to that of the cone, and at $20480$
layers the cone is matched to $1\times10^{-9}$ for $p\leq8$, whereas flat space
misses by $8$--$18\%$. The free propagator converges to that of
the cone at the same order, and heat diffusion reproduces the $r^6$ measure and the
Bessel spectrum of the cone. The maps between neighbouring layers are obtained in closed form from the polar decomposition of the
second-quantized Schwinger-boson lift, using a known operator identity, and coincide
with the Kov\'a\v{c}ik--Tekel coefficient-copying prescription. For $\CP^1$ the
construction reduces, away from the outer boundary, to the Kov\'a\v{c}ik--Tekel
fuzzy onion. Modes along the twistor fibre are exactly degenerate with the
harmonics on $S^4$.
\end{abstract}

\subjectindex{B80, B82, B83}

\maketitle

\section{Introduction}
\label{sec:intro}

Matrix models provide some of the most concrete proposals for a nonperturbative
formulation of string theory and quantum gravity, most prominently the BFSS
model~\cite{BFSS} and the IKKT model~\cite{IKKT}. In these models no background
spacetime is assumed: macroscopic geometry emerges from configurations of large matrices, and
its fluctuations are governed by the matrix dynamics~\cite{Steinacker2010,KNT}.
The simplest emergent geometries are fuzzy spaces~\cite{Hoppe,Madore}, in which a
compact manifold is replaced by a finite-dimensional matrix algebra with an intrinsic
minimal length. Fuzzy spheres and fuzzy complex projective spaces arise as
classical solutions of matrix models with cubic terms~\cite{Myers,IKTW,ABNN,Kitazawa},
and field theories on them are finite by construction. The standard examples, the
fuzzy sphere and the fuzzy complex projective spaces~\cite{BDLMO,ABIY}, are
quantized coadjoint orbits: they inherit a symplectic structure from the underlying
Lie algebra and are therefore necessarily even-dimensional. Odd-dimensional fuzzy
spheres can be defined, but only through considerably more elaborate
constructions~\cite{GR,Ramgoolam,Sako}. A space with a genuine radial direction can
instead be assembled from fuzzy spheres of increasing size. This is the structure
of the noncommutative space $\mathbb{R}^3_\lambda$~\cite{HLS,VW}, which also follows
from the Moyal space $\mathbb{R}^4$ by reduction along the Jordan--Schwinger
map~\cite{GLMV}, and in the
rotationally invariant realization of G\'alikov\'a and Pre\v{s}najder on the Fock space of
two Schwinger bosons~\cite{GP,GKP} the Laplacian couples neighbouring spheres; its
phenomenology has been studied in~\cite{KDR}. Foliations of $\mathbb{R}^3$ by fuzzy
spheres and of $\mathbb{R}^4$ by fuzzy three-spheres also arise from quantized
Nambu--Poisson structures~\cite{DSS}, and a cone, the conifold, has been quantized
on the Fock space of four oscillators~\cite{AV}.

A matrix-model formulation of this idea was proposed by Kov\'a\v{c}ik
and Tekel~\cite{KTcorfu,KT}. Their fuzzy onion stacks concentric fuzzy spheres of
increasing matrix size and reinterprets the matrix size itself as a radial
coordinate. Fields become block-diagonal matrices; the angular Laplacian acts
within each block, while the radial derivative becomes a finite difference between
neighbouring blocks, implemented by maps that copy expansion coefficients from one
layer to the next. The result is a three-dimensional quantum space, with two
angular directions and one radial direction, a regular centre and a boundary, on
which scalar field theory can be studied with the tools of matrix
models~\cite{KTH,HKT}.

Extending this idea to higher dimensions raises a basic question about the
continuum limit. A discrete stack of matrix layers may converge to flat space or to
a curved geometry such as the metric cone over the layer. The metric cone $C(X)$ over a manifold $X$ is
$\mathbb{R}_{>0}\times X$ with the metric $dr^2+r^2g_X$, whose section at radius $r$
is $X$ scaled by $r$~\cite{Cheeger}; $X$ is called the link of the cone. If $X$ is the round unit sphere $S^n$, $C(X)$
is flat $\mathbb{R}^{n+1}$ (minus the origin); otherwise the tip $r=0$ is singular.
For the fuzzy sphere the two alternatives therefore coincide, $C(S^2)=\mathbb{R}^3$,
and the original onion cannot distinguish between them.
For $\mathbb{CP}^n$ with $n\geq2$ this degeneracy is lifted, and which geometry the
matrix model selects becomes a meaningful question. We consider $\mathbb{CP}^3$. It
is the twistor space of the four-sphere~\cite{Penrose,AHS} and underlies both the
construction of fuzzy four-spheres and the higher-spin structure of covariant
quantum spaces in matrix
models~\cite{CLT,MedinaOConnor,MHOD,SperlingSteinacker,Steinacker}; it also appears
as the internal space in the gravity dual of ABJM theory~\cite{ABJM}. Whether a
stack of $\mathbb{CP}^3$ layers converges to flat $\mathbb{R}^7$ or to
$C(\mathbb{CP}^3)$ has not been determined.

In this paper we construct the fuzzy $\mathbb{CP}^3$ onion, a seven-dimensional
quantum space built from concentric fuzzy $\mathbb{CP}^3$'s, and use it to answer
this question. Our main result is that its continuum limit is the metric cone
$C(\mathbb{CP}^3)$, not flat $\mathbb{R}^7$. Because every ingredient is fixed in
advance, this is a prediction of the construction, with no fitted parameters. In each fixed angular
sector (harmonics on $\CP^3$ of degree $p$) the spectrum converges at second order to that of the cone; the two
candidates differ in every angular sector beyond the $s$ wave ($p=0$). At $20480$ layers the
cone is matched to $1\times10^{-9}$ in the sectors $p\leq8$, while flat space misses
by $8$--$18\%$ for $1\leq p\leq8$. The free propagator of the scalar field converges
to that of the cone as well, sector by sector and, with the fields represented by
their symmetric symbols, in position space, and differs from the flat one by
$21$--$37\%$ in the lowest sectors beyond the $s$ wave. Heat diffusion expresses the same
geometry in thermodynamic terms: the equilibrium temperature reproduces the $r^6$ growth of
the measure, and the relaxation times of the angular sectors the Bessel zeros of the
cone. In its continuum limit the construction thus differs both from the
three-dimensional onion~\cite{KT}, whose limit is flat, and from the fuzzy
$S^7$~\cite{Sako}, which is seven-dimensional but compact and is constructed by a
different method.

To reach this conclusion the onion must first be constructed on $\mathbb{CP}^3$.
The construction of~\cite{KT} rests on features specific to the two-sphere. The
copy prescription presupposes that a given harmonic on one layer and its
counterpart on the next are the same mode. This premise is transparent for
spherical harmonics, whose labels are the same on every layer, but it has no
evident counterpart on a general coadjoint orbit: the central obstacle is the
absence of a canonical map between layers of different size. We
resolve it without choosing a basis at all: realizing each layer on the Fock space
of four Schwinger bosons, we obtain the interlayer maps in closed form from the polar
decomposition of the second-quantized lift, which is the embedding of fuzzy $\CP^N$
harmonics used by Dolan, Huet, Murray and O'Connor~\cite{DHMO} and whose polar
factor follows from their operator identity. The Kov\'a\v{c}ik--Tekel prescription,
carried out in an explicitly constructed tensor basis, reproduces these maps
identically; this fixes their normalization once ``the same mode on two layers'' is
understood, as in~\cite{KT}, as equality of expansion coefficients. The radial
Laplacian then follows from the integration rule of fuzzy geometry, with the radial
factor of each bond evaluated at its geometric midpoint, and away from the outer
boundary it reduces exactly to that of~\cite{KT} for the two-sphere. The completely
positive normalization obtained by dividing the same lift by the layer index
reproduces instead, for the two-sphere, the Laplacian of~\cite{GP}, so that both
three-dimensional constructions are special cases of the present one. With this
normalization the limit is the same cone, but it is reached only at first order. The contribution of this work is therefore not
the operator identity, which is known~\cite{DHMO}, but its identification with
coefficient copying, the resulting free radial operator, its convergence to the cone
and the comparison of the propagators.

Modes along the twistor fibre are degenerate with the harmonics on $S^4$ and, seen
from $S^4$, behave as massless higher-spin fields. The $Spin(5)$-invariant
squashings~\cite{MHOD} that make the fibre visible and can be obtained by rescaling
the $\mathfrak{su}(4)$ generators cannot bring the fibre-to-base radius ratio below
$1/\sqrt2$, and they open no five-dimensional energy window. The model is constructed
rather than derived: the radial kinetic term is not obtained from a parent matrix
model, although concentric configurations of fuzzy spheres have recently been found
to be long-lived in a three-matrix model~\cite{Hrmo2026}, and the continuum limit is
established numerically, mode by mode.

The paper is organized as follows. Section~\ref{sec:cp3} introduces fuzzy
$\mathbb{CP}^3$ through its Schwinger-boson realization, including the
multiplicity-free harmonic decomposition and the trace. Section~\ref{sec:onion}
constructs the fuzzy $\mathbb{CP}^3$ onion, including the layer radii, the
interlayer maps, their uniqueness and the radial Laplacian.
Section~\ref{sec:continuum} identifies the continuum limit of the spectrum and of the
free propagator,
Section~\ref{sec:heat} studies heat diffusion, and Section~\ref{sec:kk} examines
the modes of the twistor fibre and its squashing. Section~\ref{sec:conclusion}
summarizes the results and outlines future directions. The effect of deforming the Up map
(the map from layer $j$ to layer $j+1$) is examined in Appendix~\ref{app:centrifugal}, and the linear program for the squashing
is described in Appendix~\ref{app:lp}.

\section{Fuzzy \texorpdfstring{$\mathbb{CP}^3$}{CP3}}
\label{sec:cp3}

\label{sec:S2}

The fuzzy sphere~\cite{Madore} truncates the spherical-harmonic expansion at
$l\leq N-1$ and replaces the pointwise product by the matrix product; the
$\sum_{l<N}(2l+1)=N^2$ modes match the degrees of freedom of an $N\times N$
matrix. The matrix spherical harmonics $Y^{(N)}_{lm}$ satisfy
\begin{equation}\label{eq:S2harm}
[L_i,[L_i,Y^{(N)}_{lm}]]=l(l+1)\,Y^{(N)}_{lm}
\end{equation}
with $L_i$ the generators of the spin-$\frac{N-1}{2}$ representation~\cite{KT}.

\subsection{Schwinger-boson realization}
\label{sec:bosons}
\label{sec:boson}

The isometry group of $\CP^3$ is $SU(4)$. The fuzzy $\CP^3$ at level
$k=0,1,2,\dots$ is the matrix algebra $\End(V_k)$ of linear operators on the totally symmetric
representation $V_k=\mathrm{Sym}^k(\mathbb{C}^4)$~\cite{BDLMO}. To work with explicit
matrices we realize $V_k$ with four species of creation and annihilation operators,
\begin{equation}\label{eq:ccr}
[a_i,a_j^\dagger]=\delta_{ij},\qquad i,j=1,\dots,4.
\end{equation}
The eigenspace of the number operator $\hat N=\sum_ia_i^\dagger a_i$ with eigenvalue
$k$ (the $k$-particle subspace) is precisely $V_k$. The onion labels its layers by
$j=k+1$ (layer $j$ is the fuzzy $\CP^3$ at level $k=j-1$); in terms of this index the
matrix size is
\begin{equation}\label{eq:Nj}
N_j=\dim V_{j-1}=\binom{j+2}{3}=\frac{j(j+1)(j+2)}{6}=1,\,4,\,10,\,20,\,35,\,\dots .
\end{equation}
The $SU(4)$ generators
are
\begin{equation}\label{eq:gens}
E_{ij}=a_i^\dagger a_j,\qquad \rho(T^A)=\sum_{i,j}(T^A)_{ij}\,E_{ij},
\end{equation}
where $T^A$ ($A=1,\dots,15$) are the $\mathfrak{su}(4)$ generators in the
fundamental representation, normalized as
$\tr(T^AT^B)=\frac12\delta^{AB}$.\footnote{The $T^A$ are $4\times4$ matrices,
whereas $\rho(T^A)$ is their $N_j\times N_j$ image on layer $j$; the two coincide
only for $j=2$. The normalization is not preserved by $\rho$:
$\tr_{V_k}\bigl(\rho(T^A)\rho(T^B)\bigr)=I_2(k)\,\delta^{AB}$ with
$I_2(k)=\frac12\binom{k+4}{5}$. The angular Laplacian and
$\hat C_2$ are built from $\rho(T^A)$.} With the noncommutativity length scale $\lambda$ and coordinate operators
$X^A=2\lambda\,\rho(T^A)$ one finds
\begin{equation}\label{eq:XA2}
\sum_A(X^A)^2=\frac32\lambda^2(k+2)^2-6\lambda^2,
\end{equation}
so that for large $k$ the layer lies on a sphere of radius
$\sqrt{3/2}\,\lambda(k+2)$ in $\mathbb{R}^{15}$. We write the coherent
state~\cite{Perelomov} of the point $[z]\in\CP^3$ ($z\in\mathbb{C}^4$, $|z|=1$) as
\begin{equation}\label{eq:zstate}
|z\rangle:=\frac{1}{\sqrt{k!}}\Bigl(\sum_{i=1}^4z_ia_i^\dagger\Bigr)^{k}|0\rangle\in V_k ;
\end{equation}
up to a phase it depends only on $[z]$, near which it is localized.

\subsection{Harmonic expansion and integration}
\label{sec:cp3harm}
\label{sec:why}

A field on layer $j$ is an $N_j\times N_j$ Hermitian matrix $\Phi^{(j)}$, i.e.\ an
operator on the Fock layer $V_{j-1}$,
\begin{equation}\label{eq:endtensor}
\End(V_{j-1})=(j-1,0,0)\otimes(0,0,j-1),\qquad \dim=N_j^2,
\end{equation}
where irreducible representations of $SU(4)$ are denoted by Dynkin labels $(a,b,c)$
and $\overline{(a,b,c)}=(c,b,a)$. Under $SU(4)$ this space decomposes
without multiplicity~\cite{BDLMO,DHMO},
\begin{equation}\label{eq:harm}
\End(V_{j-1})\cong\bigoplus_{p=0}^{j-1}(p,0,p),
\quad
d_p:=\dim(p,0,p)=\frac{(p+1)^2(p+2)^2(2p+3)}{12}=1,\,15,\,84,\,300,\,\dots,
\end{equation}
each irreducible component occurring exactly once. The continuum counterpart is the Cartan--Helgason theorem~\cite{Helgason} for the Hermitian
symmetric space $\CP^3=SU(4)/U(3)$, which states that
$L^2(\CP^3)=\bigoplus_{p\geq0}(p,0,p)$ with multiplicity one, of which~\eqref{eq:harm} is the truncation
to $p\leq j-1$. We use two facts below. First, elements of $\End(V_{j-1})$ are sums
of normal-ordered monomials $(a^\dagger)^\alpha a^\beta$, with multi-indices $\alpha,\beta$,
$(a^\dagger)^\alpha=\prod_i(a_i^\dagger)^{\alpha_i}$ and $|\alpha|=\sum_i\alpha_i=|\beta|=p$;
since $a^\beta$ removes $p$ particles, terms with $p>j-1$ vanish, which is the
origin of the cutoff $p\leq j-1$ on layer $j$. Second, the component $(p,0,p)$ is
spanned by traceless coefficient tensors $C_{\alpha\beta}$ and has the telescoping
dimension
\begin{equation}\label{eq:dptele}
d_p=N_{p+1}^2-N_p^2 ,
\end{equation}
a consequence of the commutativity of the $a_i^\dagger$.

The decomposition~\eqref{eq:harm} is the $\CP^3$ analogue of the harmonic
expansion, with $p$ playing the role of $l$:
\begin{equation}\label{eq:expansion}
\Phi^{(j)}=\sum_{p=0}^{j-1}\sum_{m=1}^{d_p}c^{(j)}_{pm}\,Y^{(j)}_{pm},
\end{equation}
where $\{Y^{(j)}_{pm}\}_m$ is a basis of $(p,0,p)$. Layer $j$
accommodates only modes with $p\leq j-1$; as for the fuzzy sphere, smaller matrices
have coarser angular resolution. Counting modes gives the increment identity
\begin{equation}\label{eq:increment}
\sum_{p=0}^{j-1}d_p=N_j^2
\end{equation}
(e.g.\ $1+15+84=100=10^2$), which is the telescoping sum of~\eqref{eq:dptele}.
Multiplicity-freeness is a structural input rather than a convenience and is used
repeatedly in the construction below. The quadratic Casimir, the analogue of $l(l+1)$, i.e.\ the eigenvalue of
$\hat C_2O:=\sum_A[\rho(T^A),[\rho(T^A),O]]$ on $(p,0,p)$, takes the values
\begin{equation}\label{eq:C2}
C_2(p)=p(p+3)=0,\,4,\,10,\,18,\,\dots .
\end{equation}

\paragraph{The trace as an integral}
\label{sec:trace}

Following the standard dictionary of fuzzy geometry~\cite{KT}, the
integral over layer $j$, with $dV$ the volume element of $\CP^3$, is represented by the trace,
\begin{equation}\label{eq:intrule}
\int_{\CP^3}f\,dV\;\longrightarrow\;\frac{\VolCP(\CP^3)}{N_j}\,\tr\,\Phi^{(j)},
\qquad
\VolCP(\CP^3)=\frac{(4\pi)^3}{3!}=\frac{32\pi^3}{3}.
\end{equation}
Setting $f\equiv1$ gives $(\VolCP/N_j)\tr\,\id_{N_j}=\VolCP$ on every layer: the
identity matrix represents the constant function $1$ irrespective of the layer.
This elementary fact will fix the normalization of the interlayer maps
(Section~\ref{sec:unique}). Moreover, $\tr\,Y^{(j)}_{pm}=0$ for $p>0$, since the
trace is an $SU(4)$-invariant linear functional and the only invariant subspace
is $p=0$.

The Fubini--Study metric $g_{\rm FS}$ carries an overall scale. In the common normalization
(holomorphic sectional curvature $4$, with $\CP^1$ a sphere of radius $1/2$) one
has $\VolCP=\pi^n/n!$ and Laplace spectrum $4p(p+n)$. Since the angular Laplacian
below has spectrum $C_2(p)=p(p+n)$, our metric is $\tilde g=4g_{\rm FS}$, with
volume
\begin{equation}\label{eq:volnorm}
\VolCP(\CP^n)=\frac{(4\pi)^n}{n!}.
\end{equation}
For $n=1$ this is the area $4\pi$ of the unit sphere, reproducing the measure
$\mathcal V_1=4\pi\lambda^3$ of~\cite{KT}. The volume enters only as an
overall constant.

\section{The fuzzy \texorpdfstring{$\mathbb{CP}^3$}{CP3} onion}
\label{sec:onion}

\subsection{Layers, fields and the angular Laplacian}
\label{sec:radius}

The onion stacks the fuzzy $\CP^3$'s $\End(V_{j-1})$, $j=1,\dots,M$, as concentric
layers. We assign to layer $j$ the \emph{dimension-based radius}
\begin{equation}\label{eq:rdim}
r_j=\lambda\,(6N_j)^{1/3}=\lambda\,[\,j(j+1)(j+2)\,]^{1/3}
\;\xrightarrow{\ j\gg1\ }\;\lambda(j+1).
\end{equation}
The motivation is the Weyl rule of geometric quantization: when $\CP^n$, of real
dimension $2n$, is fuzzified at scale $\lambda$, the number of modes should be
proportional to the volume $\propto r^{2n}$ divided by the cell volume
$\propto\lambda^{2n}$, so that $N_j^2\propto(r_j/\lambda)^{2n}$. The Weyl rule fixes
only this proportionality. We choose the constant so that for large $j$ the layers
are spaced by the step $\lambda$ ($r_j\simeq\lambda j$, since $N_j\simeq j^n/n!$).
Setting $N_j^2=(r_j/\lambda)^{2n}/(n!)^2$ then gives
\begin{equation}\label{eq:rn}
r_j^{\,n}=\lambda^n\,n!\,N_j ,
\end{equation}
which for $n=3$ is~\eqref{eq:rdim}. For $n=1$ it gives $r_j=\lambda j$, the
radius of~\cite{KT}. Other assignments, such as the uniform lattice
$r_j=\lambda(j+1)$ or the Casimir radius~\eqref{eq:XA2}, which is
$\lambda\sqrt{k(k+4)}$ ($k=j-1$) up to the overall factor $\sqrt{3/2}$, agree with~\eqref{eq:rdim} at leading order in $1/j$.
We adopt~\eqref{eq:rdim} because the radial measure then coincides exactly with the
mode count (see~\eqref{eq:radint} below).

As in~\cite{KT}, the radial integral is defined as a Riemann
sum,
\begin{equation}\label{eq:radint}
\int d^7x\,\psi=\int r^{6}\,dr\int_{\CP^3}d\Omega\,\psi
\;\longrightarrow\;
\sum_j r_j^{6}\,\lambda\,\frac{\VolCP}{N_j}\,\tr\,\Phi^{(j)}
=36\,\VolCP\,\lambda^{7}\sum_j N_j\,\tr\,\Phi^{(j)},
\end{equation}
where the last equality uses $r_j^6=\lambda^6(6N_j)^2$. The weight of layer $j$ is
therefore $w_j\propto N_j$, and including the trace of the identity matrix the radial
measure $h_j=w_jN_j=N_j^2$ equals exactly the number of modes of layer $j$. This is
consistent with the increment identity~\eqref{eq:increment}: advancing by one layer increases the measure by
the number of newly opened modes.

\paragraph{Fields as block-diagonal matrices}
\label{sec:fields}

A field on an $M$-layer onion is a block-diagonal matrix
\begin{equation}\label{eq:blockdiag}
\Psi=\mathrm{diag}\bigl(\Phi^{(1)},\Phi^{(2)},\dots,\Phi^{(M)}\bigr),
\qquad \sum_{j=1}^M N_j=\binom{M+3}{4}
\end{equation}
(cf.~\cite{KT}). By~\eqref{eq:radint}, the inner product (the integral
of the product of two fields) is, up to an overall constant, given by $G$:
\begin{equation}\label{eq:metric}
\langle\Psi,\Psi'\rangle_G=\sum_{j=1}^{M}N_j\,\tr\bigl(\Phi^{(j)}\Phi'^{(j)}\bigr).
\end{equation}
The number of degrees of freedom is $\sum_{j=1}^{M}N_j^2\simeq M^7/252$: one
radial direction with $M$ layers and six angular directions with $\sim M^6$ modes
per layer. For $M\to\infty$ the fields $\bigoplus_j\End(V_{j-1})$ are the operators
on the Fock space of the four bosons that preserve the particle number. With two
bosons the same algebra is $\mathbb{R}^3_\lambda$, which is obtained in this way from
the Moyal space $\mathbb{R}^4$~\cite{GLMV,GP}.

\paragraph{The angular Laplacian}
\label{sec:KL}

Each layer carries the angular Laplacian of fuzzy $\CP^3$,
\begin{equation}\label{eq:KL}
K_L\Phi^{(j)}=\frac{1}{r_j^2}\sum_{A=1}^{15}\bigl[\rho(T^A),[\rho(T^A),\Phi^{(j)}]\bigr]
=\frac{1}{r_j^2}\,\hat C_2\,\Phi^{(j)},
\qquad
\hat C_2\,Y^{(j)}_{pm}=p(p+3)\,Y^{(j)}_{pm},
\end{equation}
the direct analogue of $[L_i,[L_i,\cdot\,]]$ in~\cite{KT}. Since $K_L$
is a positive operator (representing $-\Delta_{\CP^3}$), the negative Laplacian
of the onion is
\begin{equation}\label{eq:Delta}
\Delta:=K_R-K_L ,
\end{equation}
with the radial part $K_R$ constructed in Section~\ref{sec:KR}. Restricted to
sector $p$ it reads $K_R-C_2(p)/r^2$. Since the cutoff $p\leq j-1$ grows with the
radius, sector $p$ exists only on layers $j\geq p+1$, and its radial problem lives
on the truncated chain $j=p+1,\dots,M$.

\subsection{Interlayer maps}
\label{sec:updown}

The central ingredient of the onion is a map comparing fields on neighbouring
layers of different size. We call a map carrying fields from layer $j$ to layer
$j+1$ an Up map, and one in the opposite direction a Down map. For $S^2$, Kov\'a\v{c}ik and Tekel defined it by copying
expansion coefficients, padding missing modes with zeros and discarding surplus
modes~\cite{KT}. For $\CP^3$ explicit bases of $(p,0,p)$
($d_p=300,825,\dots$) are unwieldy, and we first give a basis-free construction.

The boson operators define a natural map between layers,
\begin{equation}\label{eq:U2q}
U_{2q}:\End(V_{j-1})\to\End(V_{j}),\qquad
U_{2q}(O)=\sum_{i=1}^4 a_i^\dagger\,O\,a_i ,
\end{equation}
from layer $j$ to layer $j+1$, which adds one particle on each side of $O$ and manifestly commutes with $SU(4)$;
it is the embedding by which~\cite{DHMO} build the polarization tensors
of fuzzy $\CP^N$.
Since~\eqref{eq:harm} is multiplicity-free, Schur's lemma implies that $U_{2q}$
acts on each $(p,0,p)$ block as a scalar $s_p(j)$ with respect to
Hilbert--Schmidt (HS) orthonormal bases. The identities
below are in~\cite{DHMO} (the second follows from the first
and~\eqref{eq:C2}; $U_{2q}^\dagger$ is the adjoint with respect to the HS inner product):
\begin{equation}\label{eq:polarid}
U_{2q}^\dagger U_{2q}=j(j+3)\,\id-\hat C_2 ,
\qquad
s_p(j)^2=(j-p)(j+p+3).
\end{equation}
Because $s_p(j)$ depends on $p$, $U_{2q}$ itself is not a suitable Up map, but the
first identity gives its polar decomposition explicitly:
\begin{equation}\label{eq:Uhat}
\Uhat=U_{2q}\,\bigl[j(j+3)-\hat C_2\bigr]^{-1/2}.
\end{equation}
This is an $SU(4)$-equivariant isometry acting on each block as a positive scalar,
i.e.\ the canonical identification of the $(p,0,p)$ block on layer $j$ with the
same block on layer $j+1$.

\paragraph{Normalization and uniqueness of the Up map}
\label{sec:unique}

By multiplicity-freeness and Schur's lemma, a general $SU(4)$-equivariant
comparison map $B_j$ is specified by one positive number per block,
\begin{equation}\label{eq:tp}
B_j\bigl(e^{(j)}_{p}\bigr)=t_p(j)\,e^{(j+1)}_{p},
\end{equation}
where $e^{(j)}_p$ denotes an HS-normalized element of the $(p,0,p)$ block. We
consider the following three conditions on $B_j$:
\begin{center}
\begin{tabular}{l@{\qquad}l@{\qquad}l}
(i) & equivariance & $B_j$ is $SU(4)$-equivariant \\[2pt]
(ii) & conformality & $t_p(j)=t(j)$ is independent of $p$ \\[2pt]
(iii) & unitality & $B_j(\id_{N_j})=\id_{N_{j+1}}$
\end{tabular}
\end{center}
Unitality expresses the fact that the identity matrix
represents the constant function $1$ on every layer (Section~\ref{sec:trace}).
Since $\Uhat$ is an isometry, $\Uhat(\id_{N_j})=\sqrt{N_j/N_{j+1}}\,\id_{N_{j+1}}$,
and (i)--(iii) determine the Up map uniquely:
\begin{equation}\label{eq:Utilde}
\Util:=\sqrt{\frac{N_{j+1}}{N_j}}\;\Uhat,\qquad \Util(\id)=\id .
\end{equation}
Only this normalization makes the radial energy of the uniform configuration
$\Phi\propto\id$ vanish exactly; any other constant factor leaves a term of order
$\lambda^{-2}$ and precludes a continuum limit.

Unitality, however, constrains only the $p=0$ sector: the identity lies in the
$p=0$ block, and the trace annihilates all harmonics with $p>0$. Unitality,
conservation of the layer integral, and vanishing of the radial energy of the
uniform configuration are therefore statements about $p=0$ alone and leave
$t_1,t_2,\dots$ unconstrained. Conformality (ii) is not an independent
assumption, however; it follows from Theorem~\ref{thm:copymain} below, i.e.\ from
the coefficient-copying premise of~\cite{KT}.

The key is to give a precise meaning, on $\CP^3$, to the premise
of~\cite{KT} that the $m$-th mode on layer $j$ and the $m$-th mode on
layer $j+1$ are the same. On $S^2$ a spherical harmonic is a polynomial
$C_{i_1\cdots i_l}x_{i_1}\cdots x_{i_l}$ with a traceless symmetric coefficient
tensor $C$; on the fuzzy sphere one replaces $x_i$ by $L_i$, and the same $C$
labels the same mode on every layer. On $\CP^3$ this role is played by
polynomials in the Schwinger bosons: for multi-indices $\alpha,\beta$ with
$|\alpha|=|\beta|=p$ and traceless coefficients $C_{\alpha\beta}$, let
\begin{equation}\label{eq:Ytensor}
Y^{(j)}(C):=\sum_{|\alpha|=|\beta|=p}C_{\alpha\beta}\,(a^\dagger)^\alpha a^\beta\Big|_{V_{j-1}},
\qquad j\geq p+1 .
\end{equation}
Then $Y^{(j)}(C)$ lies in the $(p,0,p)$ block of every layer,\footnote{%
The map $C\mapsto Y^{(p+1)}(C)$ is an $SU(4)$-equivariant isomorphism onto
$\End(V_p)$, and the trace map contracting one pair of indices,
$\mathcal T:\mathrm{Sym}^p(\mathbb{C}^4)\otimes\mathrm{Sym}^p(\mathbb{C}^4)^*\to\mathrm{Sym}^{p-1}(\mathbb{C}^4)\otimes\mathrm{Sym}^{p-1}(\mathbb{C}^4)^*$,
is surjective, so $\dim\ker\mathcal T=N_{p+1}^2-N_p^2=d_p$ and by multiplicity one the traceless
$C$ span exactly $(p,0,p)$. By~\eqref{eq:Cup} below, $Y^{(j)}(C)\neq0$ for
$j\geq p+1$.} and elements with the same $C$ can be regarded as ``the same
mode.'' The construction of $\Util$ itself uses no basis; a basis enters only
for comparison with the prescription of~\cite{KT}, which is formulated in terms
of expansion coefficients.

\begin{theorem}[The copy prescription reproduces $\Util$]
\label{thm:copymain}
Normalize the basis $\hat Y^{(j)}_{pm}$ of $(p,0,p)$ on each layer by the layer
integral~\eqref{eq:intrule}, as
in~\cite{KT}, and let $U_{\rm basis}$ be the map that copies
coefficients in this basis and pads missing modes with zeros. Then
$U_{\rm basis}=\Util$ for all $j$ and $p$.
\end{theorem}

\begin{proof}
Since $\sum_ia_i^\dagger(a^\dagger)^\alpha a^\beta a_i=(a^\dagger)^\alpha\hat Na^\beta$ and
$a^\beta$ removes $p$ particles on $V_j$,
\begin{equation}\label{eq:Cup}
U_{2q}\,Y^{(j)}(C)=(j-p)\,Y^{(j+1)}(C).
\end{equation}
Comparing norms with~\eqref{eq:polarid} for the HS-normalized
$e^{(j)}_p:=Y^{(j)}(C)/\|Y^{(j)}(C)\|_{\HS}$ gives $\Uhat e^{(j)}_p=+e^{(j+1)}_p$, the
sign following from $j-p>0$. The basis normalized by the integration rule is
$\hat Y^{(j)}_p=\sqrt{N_j}\,e^{(j)}_p$, so $U_{\rm basis}=\sqrt{N_{j+1}/N_j}\,\Uhat=\Util$;
the block $p=j$, which first appears on layer $j+1$, has no preimage, which reproduces
the zero padding.
\end{proof}

The basis-free map~\eqref{eq:Utilde} is thus the prescription of~\cite{KT}
carried out literally on $\CP^3$, including the zero padding of the top block,
and no condition of type (ii) enters.

The premise itself is a choice of what ``the same function on two layers''
means, and it is instructive to see what other choices do. Deforming
$t_p(j)=\sqrt{N_{j+1}/N_j}\,(1+\delta_p/j)$ keeps the dimension seven in the
continuum limit but shifts the centrifugal term,
$p(p+3)\to p(p+3)+\delta_p(\delta_p+5)$, in conflict with the spectrum of the
angular Laplacian~\eqref{eq:KL}, which acts within each layer
(Appendix~\ref{app:centrifugal}). Deviations of relative order $j^{-2}$, by
contrast, leave the continuum limit unchanged. An example is
\begin{equation}\label{eq:Upos}
U^{\rm pos}:=\frac1j\,U_{2q},\qquad U^{\rm pos}(\id)=\id ,
\end{equation}
which is $SU(4)$-equivariant, unital and completely positive, and preserves
exactly the symbol $\langle z|\Phi|z\rangle$ in the coherent states~\eqref{eq:zstate}.
The basis $\hat Y^{(j)}_p=\sqrt{N_j}\,e^{(j)}_p$ of Theorem~\ref{thm:copymain} is, up
to a layer-independent constant, the set of polarization tensors
of~\cite{DHMO}, whose symbols are $T_{j-1}(p)^{1/2}Y_p(z)$ with
$Y_p(z)$ independent of $j$ and $T_L(p)=L!\,(L+3)!/[(L-p)!\,(L+p+3)!]$~\cite{DHMO}. Here $T_L(p)$ is the normalization factor giving the amplitude $T_{j-1}(p)^{1/2}$ of the symbol on each layer, and it tends to $1$ as $L\to\infty$. Hence
\begin{equation}\label{eq:UposT}
U^{\rm pos}\hat Y^{(j)}_p=\Bigl[\frac{T_{j-1}(p)}{T_j(p)}\Bigr]^{1/2}\hat Y^{(j+1)}_p,
\qquad
\frac{T_{j-1}(p)}{T_j(p)}=\frac{(j-p)(j+p+3)}{j(j+3)}=1-\frac{p(p+3)}{j(j+3)},
\end{equation}
whereas $\Util$ maps $\hat Y^{(j)}_p$ to $\hat Y^{(j+1)}_p$; the ratio is $1+O(j^{-2})$.

This map is not hypothetical. For $\CP^1$ (two bosons, $N_j=j$, $r_j=\lambda j$)
the onion Laplacian built from $U^{\rm pos}$ coincides, on every layer below the
outermost, with the Laplacian
$-\lambda^{-2}(\hat N+1)^{-1}[a_\alpha^\dagger,[a_\alpha,\cdot\,]]$ of the
noncommutative space of G\'alikov\'a and Pre\v{s}najder~\cite{GP}: written
layer by layer, the double commutator is
$U_{2q}\Phi^{(j-1)}+U_{2q}^\dagger\Phi^{(j+1)}-2j\,\Phi^{(j)}$, and the term
$\hat C_2/(\lambda j)^2$ contained in $U^{{\rm pos}\dagger}U^{\rm pos}/\lambda^2$ is
exactly the angular Laplacian. The three-dimensional constructions of~\cite{GP}
and~\cite{KT} are therefore the onion Laplacians of $U^{\rm pos}$ and of $\Util$
for $\CP^1$, and the present construction contains both. For $\CP^3$ the two maps have the
same continuum limit but different convergence rates, which we compare in
Section~\ref{sec:rate}.

\paragraph{The Down map}
\label{sec:down}

Imposing the downward analogues of (i)--(iii) on the map in the opposite direction,
$D_j:\End(V_j)\to\End(V_{j-1})$ from layer $j+1$ to layer $j$, determines it as
well. By (i) and multiplicity one, $D_j$ acts on each block as a scalar, and the
top block $p=j$ of $\End(V_j)$, which has no target in $\End(V_{j-1})$, is
annihilated: the rule of~\cite{KT} that surplus modes are discarded is enforced by
representation theory. (ii) and (iii) fix the rest,
\begin{equation}\label{eq:D}
D_j=\sqrt{\frac{N_j}{N_{j+1}}}\;\Uhat_j^\dagger
=\sqrt{\frac{N_j}{N_{j+1}}}\;\bigl[j(j+3)-\hat C_2\bigr]^{-1/2}
\Bigl(\sum_{i=1}^4 a_i\,(\cdot)\,a_i^\dagger\Bigr),
\end{equation}
which satisfies $D_j\Util_j=\mathrm{id}$ and $\Util_jD_j=P_{\leq j-1}$, the
orthogonal projector removing $p=j$. The map $\sum_ia_i(\cdot)a_i^\dagger$ on the
right is the HS adjoint of $U_{2q}$, whose block scalar $s_p(j)$ vanishes precisely
at $p=j$ by~\eqref{eq:polarid}: truncation is not imposed but arises as a zero of
the block scalar.\footnote{The HS adjoint $\Util^\dagger=\frac{N_{j+1}}{N_j}D$
preserves the bare trace but is not unital. The two are dual, as maps
transporting observables (functions) and states (densities), respectively; the
pair $(\Util,D)$ transports observables.}

\subsection{The radial Laplacian}
\label{sec:KR}

We define the radial kinetic energy by a single expression, without separate
treatment of the boundaries:
\begin{equation}\label{eq:EU}
E[\Psi]=\frac{1}{\lambda^2}\sum_{j=1}^{M-1}N_j\,
\bigl\|\Phi^{(j+1)}-\Util\,\Phi^{(j)}\bigr\|^2_{\HS},
\qquad \|A\|^2_{\HS}=\tr A^\dagger A .
\end{equation}
This expression is the integration rule~\eqref{eq:radint} applied to
$(\partial_r\phi)^2$. Since $\Uhat$ is an isometry, the HS norm of the difference
carries a factor $N_{j+1}$ from layer $j+1$, and the effective bond weight is
$N_jN_{j+1}$. By~\eqref{eq:rdim}, $r_j^3=6\lambda^3N_j$, so that exactly
\begin{equation}\label{eq:midpoint}
36\,\lambda^{6}\,N_jN_{j+1}=r_j^{3}\,r_{j+1}^{3}.
\end{equation}
The right-hand side is the geometric mean of $r_j^6$ and $r_{j+1}^6$, i.e.\ $r^6$
evaluated at the midpoint of the bond. Applying the prescription
of~\eqref{eq:radint} (radial factor $r^6\Delta r$, angular factor
$\frac{\VolCP}{N}\tr$) to
$(\partial_r\phi)^2\to\bigl(\Phi^{(j+1)}-\Util\Phi^{(j)}\bigr)^2/\lambda^2$ yields
\begin{equation}\label{eq:EisInt}
\int d^7x\,\Bigl(\frac{\partial\phi}{\partial r}\Bigr)^{2}=\mathcal V\,E[\Psi],
\qquad
\mathcal V=36\,\VolCP(\CP^3)\,\lambda^{7}=384\pi^3\lambda^{7}.
\end{equation}
For general $\CP^n$ the radii~\eqref{eq:rn} make the geometric mean of the radial
measures of adjacent layers equal to $\lambda^{2n}(n!)^2N_jN_{j+1}$, and the same
cancellation occurs. The weight of a bond (a pair of adjacent layers $j$, $j+1$), i.e.\ the prefactor
$a_j=N_j$ in~\eqref{eq:EU}, thus follows from the integration rule, with the radial factor evaluated at the geometric midpoint, without matching
to $S^2$. The midpoint rule is a choice: the arithmetic mean of $r_j^6$ and
$r_{j+1}^6$ would give $a_j=(N_j^2+N_{j+1}^2)/(2N_{j+1})$, which does not close on
the dimensions of the model.

Although the prefactor $N_j$ in~\eqref{eq:EU} appears asymmetric, the effective
weight $N_jN_{j+1}$ realizes the symmetric midpoint rule~\eqref{eq:midpoint},
because the HS norm is taken on layer $j+1$. In the discretization of a
Sturm--Liouville operator with measure $w$ and conductance $a$, $\frac1w\frac{d}{dr}\bigl(a\frac{d}{dr}\bigr)$,
evaluating the conductance $a$ at the midpoint of the bond yields second-order
accuracy, whereas evaluation at an endpoint yields first-order
accuracy. How these choices affect the order of convergence is compared in
Section~\ref{sec:rate}.

The radial Laplacian is minus one half of the gradient $\nabla_{\!G}E$ of $E$ with
respect to the inner product $G$ of~\eqref{eq:metric}, $K_R=-\frac12\nabla_{\!G}E$. We write the bond
difference as
\begin{equation}\label{eq:bonddiff}
\Delta_j:=\Phi^{(j+1)}-\Util_j\Phi^{(j)}\in\End(V_j).
\end{equation}

Varying $\delta\Delta_j=\delta\Phi^{(j+1)}-\Util_j\,\delta\Phi^{(j)}$ and shifting the
summation index (a discrete integration by parts) gives
\begin{equation}\label{eq:gradE}
(\nabla_{\!G}E)^{(j)}=\frac{2}{\lambda^{2}}
\Bigl[\frac{N_{j-1}}{N_j}\,\Delta_{j-1}-\Util_j^{\dagger}\Delta_j\Bigr],
\end{equation}
where $\Util_j^\dagger$ is the HS adjoint of $\Util_j$, and $\Delta_0$ (removed by
$N_0=0$) and $\Delta_M$ are absent.

Thus $\Util^\dagger$ is not inserted by hand; the apparent asymmetry between
inward and outward bonds is absorbed by the metric, with respect to which $K_R$
is exactly self-adjoint; and the boundary conditions follow from the index shift:
layer $j=1$ has no inward bond and layer $j=M$ has no outward bond.
Equation~\eqref{eq:gradE} gives the \emph{flux form}
\begin{equation}\label{eq:KRflux}
\bigl(K_R\Psi\bigr)^{(j)}=\frac{1}{\lambda^{2}}\Bigl[
\Util_j^{\dagger}\bigl(\Phi^{(j+1)}-\Util_j\Phi^{(j)}\bigr)
-\frac{N_{j-1}}{N_j}\bigl(\Phi^{(j)}-\Util_{j-1}\Phi^{(j-1)}\bigr)\Bigr],
\end{equation}
in which the two brackets are the heat fluxes through the outward and inward
bonds. On the innermost layer $N_0=0$ removes the inward term, and on the
outermost layer the absence of an outward bond imposes a Neumann
condition.\footnote{Setting instead $\Phi^{(M+1)}:=0$ reproduces the Dirichlet
convention of~\cite{KT}.} At the inner end of the sector-$p$ chain,
$\Util_{p}\Phi^{(p)}$ has no $p$ component, leaving the mass term
$-\frac{N_p}{N_{p+1}}\Phi^{(p+1)}$. This \emph{ghost term} agrees with the
zero-padding convention of~\cite{KT}.

Using $\Util_j^\dagger\Util_j=\frac{N_{j+1}}{N_j}\id$ and
$\Util_j^\dagger=\frac{N_{j+1}}{N_j}D_j$, the flux form becomes
\begin{equation}\label{eq:KRpsi}
\bigl(K_R\Psi\bigr)^{(j)}=\frac{1}{\lambda^{2}}\left[
\frac{N_{j+1}}{N_j}\,D_j\Phi^{(j+1)}
+\frac{N_{j-1}}{N_j}\,\Util_{j-1}\Phi^{(j-1)}
-\frac{N_{j+1}+N_{j-1}}{N_j}\,\Phi^{(j)}\right],
\end{equation}
which corresponds term by term to~\cite{KT},
\begin{equation}\label{eq:KT27}
K_R\Psi=\sum_{N,l,m}\frac{(N+1)c^{(N+1)}_{lm}+(N-1)c^{(N-1)}_{lm}-2N\,c^{(N)}_{lm}}{N\lambda^2}\,Y^{(N)}_{lm},
\end{equation}
upon setting $N_j=j$ and identifying $D$ and $\Util$ with coefficient copying
($N$ is the matrix size of a layer and $c^{(N)}_{lm}$ the expansion coefficients).
The correspondence holds on every layer except the outermost, where~\cite{KT}
set $U\Phi^{(M)}=0$, i.e.\ impose a Dirichlet condition, whereas~\eqref{eq:KRflux}
has no outward bond and imposes a Neumann condition; for heat
transfer~\cite{KT} also use a Neumann condition. With $U^{\rm pos}$ in place of
$\Util$ the same expression reproduces, for $N_j=j$, the Laplacian of~\cite{GP}
(Section~\ref{sec:unique}).
The fuzzy onion of~\cite{KT} and the fuzzy $\CP^3$ onion are thus the cases $n=1$
and $n=3$ of a single construction with $n+1$ bosons, for which the identities of
Section~\ref{sec:updown} hold with $j(j+3)\to j(j+n)$~\cite{DHMO}. For $n=1$ several
choices are degenerate, because $N_j^2=j^2=j^{D-1}$, with $D=2n+1$ the dimension,
holds identically; this is why
matching to~\cite{KT} alone could not fix the $\CP^3$ construction.
For the coefficients $c_j$ of a single sector ($\Phi^{(j)}=c_j\id$ for $p=0$, and
the coefficient of a basis element transported by $\Util$ for $p>0$),
\begin{align}\label{eq:KR}
(K_R\,c)_j&=\frac{1}{\lambda^2}\left[
\frac{N_{j+1}}{N_j}\,c_{j+1}+\frac{N_{j-1}}{N_j}\,c_{j-1}
-\frac{N_{j+1}+N_{j-1}}{N_j}\,c_j\right]\nonumber\\
&=\frac{1}{\lambda^2}\left[
\frac{j+3}{j}c_{j+1}+\frac{j-1}{j+2}c_{j-1}-\frac{2(j^2+2j+3)}{j(j+2)}c_j\right].
\end{align}
The chain of sector $p$ obeys the same expression on $j=p+1,\dots,M$, with the
ghost term at the inner end and the Neumann condition at the outer end.

The operator~\eqref{eq:KR} is self-adjoint with respect to the weight
$h_j=N_j^2$. Conjugation by $S=\mathrm{diag}(N_j)$ symmetrizes it and reveals a
simple structure:
\begin{equation}\label{eq:sym}
\bigl(S\,K_R\,S^{-1}\bigr)_{j,j\pm1}=\frac{1}{\lambda^2}\qquad\text{for all $j$ and $p$}.
\end{equation}
The off-diagonal entries are independent of both layer and sector, and the
geometry resides entirely in the diagonal. The structural properties of
Section~\ref{sec:props} follow almost immediately from this form, which is also
the form used in the numerical analysis of Section~\ref{sec:continuum}.

The inner end of every chain is derived (through the ghost term), but at the
outer end three conventions are natural: the derived form, in which layer $M$
has full width $\lambda$ and no outward bond; a half cell, $W_M=N_M^2/2$,
with no outward bond; and zero padding, $\Phi^{(M+1)}:=0$, the Dirichlet
convention of~\cite{KT}. These describe balls of effective radius $r_M+\lambda/2$, $r_M$ and
$r_M+\lambda$, respectively. For the derived form and the half cell the operator and the measure describe
the same ball, whereas for zero padding an $O(\lambda)$ mismatch remains, because the
measure extends only to $r_M+\lambda/2$. We adopt the derived form,
since the integration rule~\eqref{eq:radint} assigns width $\lambda$ to every
layer.

\paragraph{Structural properties of $K_R$}
\label{sec:props}

Several properties of $K_R$ follow automatically from the variational
definition. Since $E$ is a symmetric quadratic form, $K_R=-\frac12G^{-1}\nabla E$
is self-adjoint with respect to $G$, and $E\geq0$ makes it negative
semi-definite. The zero mode $K_R(\id)=0$ is equivalent to $\Util(\id)=\id$
(eq.~\eqref{eq:Utilde}), and together with self-adjointness it implies
conservation of $Q=\sum_jN_j^2c_j$. Positivity of the off-diagonal entries
of~\eqref{eq:KR} in each sector chain yields a discrete maximum principle. The
latter concerns the coefficients of the sector chains; for $p>0$ these are
amplitudes of angular harmonics, and the evolution does not preserve positivity
of the matrices $\Phi^{(j)}$: a positive initial matrix can develop negative
eigenvalues (an example is given in Section~\ref{sec:localized}).

The properties with independent content are the exact reduction to the
construction of~\cite{KT} for $N_j=j$ (shown above) and the continuum
limit (Section~\ref{sec:continuum}), in which the sectors $p>0$ are tested
directly and independently of the $p=0$ normalization.

The Laplacian defines the action of a scalar field theory on the onion
(cf.~\cite{KT}),
\begin{equation}\label{eq:action}
S[\Psi]=\mathcal V\sum_{j=1}^{M}N_j\,\tr\Bigl(
-a\,\Phi^{(j)}\bigl(\Delta\Psi\bigr)^{(j)}+b\,\bigl(\Phi^{(j)}\bigr)^2+c\,\bigl(\Phi^{(j)}\bigr)^4\Bigr),
\end{equation}
with $\mathcal V$ given in~\eqref{eq:EisInt} and couplings $a,b,c$, which defines a statistical model of a
collection of Hermitian matrices. In this paper we study only the free operator
$\Delta$, i.e.\ the kinematics of the onion; interactions are not considered.

\section{Continuum limit: the metric cone over \texorpdfstring{$\mathbb{CP}^3$}{CP3}}
\label{sec:continuum}

All ingredients of Section~\ref{sec:onion} (the radii~\eqref{eq:rdim}, the bond
weight $a_j=N_j$, the measure $h_j=N_j^2$, the interlayer map $\Util$ and the
boundary terms) enter a single object, the spectrum of $\Delta$. In this section
we compute this spectrum, identify the continuum operator to which it
converges, and compare the resulting free propagator with its continuum
counterpart. The result is the Laplacian on the metric cone $C(\CP^3)$, not that of
flat $\mathbb{R}^7$.

\subsection{Radial problems by sector and candidate continuum operators}
\label{sec:reduction}

Numerical results use $\lambda=1$ unless stated otherwise. By multiplicity-freeness, $\Delta$ is block-diagonal in $p$, and within sector
$(p,0,p)$ it reduces to the tridiagonal operator
\begin{equation}\label{eq:Deltap}
-\Delta^{(p)}=K_L^{(p)}-K_R,\qquad \bigl(K_L^{(p)}\bigr)_j=\frac{p(p+3)}{r_j^2}
\end{equation}
on the chain $j=p+1,\dots,M$, with $K_R$ given by~\eqref{eq:KR}, a ghost term at
$j=p+1$ and a Neumann condition at $j=M$. By~\eqref{eq:sym}, the symmetrized operator $H^{(p)}:=-S\Delta^{(p)}S^{-1}$ is the
tridiagonal matrix with diagonal entries $H^{(p)}_{jj}=d_j$, off-diagonal (nearest-neighbour)
entries $H^{(p)}_{j,j\pm1}=e_j$ and all other entries zero, where
\begin{equation}\label{eq:tridiag}
d_j=\frac{N_{j+1}^{\rm eff}+N_{j-1}}{\lambda^2N_j}+\frac{p(p+3)}{r_j^2},
\qquad e_j=-\frac{1}{\lambda^2},
\qquad N_{M+1}^{\rm eff}:=0,\quad N_{j+1}^{\rm eff}:=N_{j+1}\ (j<M).
\end{equation}
Retaining the term $N_{j-1}=N_p$ at $j=p+1$ is the ghost term and the
absence of $N_{M+1}$ is the Neumann condition. The off-diagonal entries are
constant, and all $p$-dependence resides in the diagonal. The angular directions are solved
algebraically, leaving a one-dimensional lattice problem per sector, which we
diagonalize with a symmetric tridiagonal eigensolver (at large $M$, to avoid
round-off, we obtain the eigenvalues by shooting with the difference form of the
symmetrized chain).

\paragraph{Candidate continuum operators}
\label{sec:candidates}

The natural candidate is the metric cone $C(\CP^3)$ with metric
$dr^2+r^2\tilde g$, $\tilde g=4g_{\rm FS}$ (Section~\ref{sec:trace}). Its
Laplacian in sector $p$ is
\begin{equation}\label{eq:Lp}
L_p=\partial_r^2+\frac{6}{r}\partial_r-\frac{p(p+3)}{r^2}.
\end{equation}
With $f=r^{-5/2}g$, the equation $L_pf=-k^2f$ reduces to Bessel's equation, and
the solutions regular at the apex are, up to an arbitrary normalization,
\begin{equation}\label{eq:nup}
f(r)\propto r^{-5/2}J_{\nu_p}(kr),\qquad
\nu_p=\sqrt{\Bigl(\frac52\Bigr)^2+p(p+3)}=\sqrt{\Bigl(p+\frac32\Bigr)^2+4}.
\end{equation}
where $J_\nu$ is the Bessel function of the first kind.
Bessel functions appear because of the cone geometry itself. The metric
$dr^2+r^2\tilde g$ is only multiplied by $c^2$ under $r\to cr$, so once the Laplacian is
separated in eigenfunctions of the link, the only length scale left in the radial
equation is $1/k$. The solutions are therefore functions of $kr$, and the link
eigenvalue $p(p+3)$ fixes only the Bessel order $\nu_p$. This generalizes the
partial waves $r^{-(d-2)/2}J_{l+(d-2)/2}(kr)$ of flat $d$-dimensional space
(spherical Bessel functions for $d=3$) and holds for the cone over any
link~\cite{Cheeger}. An outer boundary at $r=R$ discretizes $k$. The outer condition of
Section~\ref{sec:KR} corresponds to the Neumann condition $f'(R)=0$, which for
$f=r^{-5/2}J_{\nu_p}(kr)$ reads
\begin{equation}\label{eq:neumann}
x\,J_{\nu}'(x)=\frac52\,J_{\nu}(x),\qquad x=kR,\quad \nu=\nu_p ,
\end{equation}
and the allowed values of $k$ are the positive roots $x$ of~\eqref{eq:neumann}
divided by $R$. We take $R=r_M+\lambda/2$, the outer edge of the last cell of width
$\lambda$ assigned by the integration rule.
Flat $\mathbb{R}^7$ is the cone over the round unit six-sphere, whose Laplacian
has spectrum $l(l+5)$. The same substitution then gives
\begin{equation}\label{eq:nuflat}
\nu_l=\sqrt{\Bigl(\frac52\Bigr)^2+l(l+5)}=l+\frac52 ,
\end{equation}
an integer-spaced tower. On the cone over $\CP^3$, $\nu_p$ is not linear in $p$, so the
tower is not integer-spaced, and since $p(p+3)<p(p+5)$ for $p\geq1$ it lies below the
flat values. The two towers \emph{coincide} at $p=0$
(both have $\nu=5/2$) and separate for every $p\geq1$, with deficits
$(p+\frac52)-\nu_p=0.298,\,0.469,\,0.576$ for $p=1,2,3$. The $s$-wave sector is
therefore blind to the difference, and a test restricted to $p=0$ would conclude
that the onion approximates flat $\mathbb{R}^7$. The difference between the
angular spectra, rather than any detail of the discretization, determines which
continuum space the onion approximates.

\paragraph{Exact Schr\"odinger form}
Since $(N_{j+1}+N_{j-1})/N_j=2+6/[j(j+2)]$ exactly, the chain~\eqref{eq:tridiag} is, in the
variable $y_j:=\lambda(j+1)$, precisely the standard three-point discretization of step
$\lambda$ of a Schr\"odinger operator, with $u_j:=N_jc_j$ the symmetrized coefficients,
\begin{equation}\label{eq:schrform}
\bigl(-\Delta^{(p)}u\bigr)_j=\frac{2u_j-u_{j+1}-u_{j-1}}{\lambda^2}+V(y_j)\,u_j,
\quad
V(y)=\frac{6}{y^2-\lambda^2}+\frac{p(p+3)}{r^2},\quad
r^2=y^2\Bigl(1-\frac{\lambda^2}{y^2}\Bigr)^{2/3},
\end{equation}
with the two ends of the chain given by the ghosts $u_p=0$ and $u_{M+1}=(1+3/M)\,u_M$
of~\eqref{eq:tridiag}. Expanding in $\lambda$,
\begin{equation}\label{eq:Vexp}
V(y)=\frac{\nu_p^2-\frac14}{y^2}+\lambda^2\,\frac{c_p}{y^4}+O\Bigl(\frac{\lambda^4}{y^6}\Bigr),
\qquad \nu_p^2-\tfrac14=p(p+3)+6,\qquad c_p=6+\tfrac23\,p(p+3),
\end{equation}
with $\nu_p$ as in~\eqref{eq:nup}. The leading term is exactly the Schr\"odinger form of
$L_p$ in~\eqref{eq:Lp}, obtained from $f=u/r^{3}$: the cone is not fitted but read off the
construction, since flat $\mathbb{R}^7$ would require the coefficient $p(p+5)+6$
by~\eqref{eq:nuflat}. The scale of the link, $\tilde g=4g_{\rm FS}$, is a convention
set by the radius assignment of Section~\ref{sec:radius} (other scales are treated in
Section~\ref{sec:linkscale}). The correction in~\eqref{eq:Vexp} is relatively of size
$\lambda^2c_p/[y^2(\nu_p^2-\frac14)]\leq\lambda^2/y^2$, uniformly in $p$; the ratio
$c_p/(\nu_p^2-\frac14)$ decreases from $1$ at $p=0$ to $\frac23$ as $p\to\infty$. This
bounds the error of the potential, not that of the three-point difference: the lowest mode of
sector $p$ has $kR\lesssim\nu_p$, so the lattice resolution $k\lambda\sim p/M$ is small only for
$p\ll M$. The expansion parameter is
$\lambda^2/y^2$, so~\eqref{eq:Vexp} says nothing about the innermost few layers.

\subsection{Mode-by-mode agreement and rate of convergence}
\label{sec:modes}
\label{sec:rate}

The continuum eigenvalues are $E^{\rm cont}_n=(x^{\rm cont}_n/R)^2$, where
$x^{\rm cont}_n$ is the $n$-th positive root of the Neumann
condition~\eqref{eq:neumann}. On the discrete side we set
$k^{\rm disc}_n:=\sqrt{E^{\rm disc}_n}$ and $x^{\rm disc}_n:=k^{\rm disc}_nR$ from the
eigenvalues $E^{\rm disc}_n$ of~\eqref{eq:tridiag}, and compare $x^{\rm disc}_n$ with
$x^{\rm cont}_n$ (a relative error $\delta$ in $x$ is a relative error $2\delta$
in $E$). The constant mode of $p=0$ ($x=0$, with eigenvalue $0$ also on the
discrete side) agrees trivially and is excluded, so $n=1$ is the first non-zero
mode; this is why $x^{\rm cont}_1$ is larger at $p=0$ in Table~\ref{tab:modes}.

\paragraph{How the limit is taken}
Since $E^{\rm disc}_n\propto\lambda^{-2}$ and $R\propto\lambda$, $x^{\rm disc}_n$
depends only on $M$, so $M\to\infty$ at $\lambda=1$ is the same as $\lambda\to0$
at fixed $R$.

Table~\ref{tab:modes} compares the two at $M=20480$ ($\lambda=1$, $R=20481.5$). The
lowest mode agrees to $1.1\times10^{-9}$ in each sector shown, and the sixth mode to
$\sim3\times10^{-8}$. No parameter is fitted: $\lambda$, the radius assignment and $R$
are fixed by Sections~\ref{sec:radius} and~\ref{sec:KR}. The comparison thus tests
not whether some discrete Laplacian converges to some continuum Laplacian, but
whether \emph{this} operator, with the bond weights and radii dictated by the
construction, converges to \emph{this} geometry.

\begin{table}[htbp]
\centering
\footnotesize
\setlength{\tabcolsep}{3.5pt}
\caption{Radial modes of the fuzzy $\CP^3$ onion at $M=20480$ compared with the
Neumann modes~\eqref{eq:neumann} of the cone $C(\CP^3)$, with $R=r_M+\lambda/2$.
The relative error is $(x^{\rm disc}_n-x^{\rm cont}_n)/x^{\rm cont}_n$.}
\label{tab:modes}
\begin{tabular}{ccccccccc}
\toprule
& \multicolumn{2}{c}{$p=0$, $\nu=2.5000$} & \multicolumn{2}{c}{$p=1$, $\nu=3.2016$}
& \multicolumn{2}{c}{$p=2$, $\nu=4.0311$} & \multicolumn{2}{c}{$p=3$, $\nu=4.9244$}\\
\cmidrule(lr){2-3}\cmidrule(lr){4-5}\cmidrule(lr){6-7}\cmidrule(lr){8-9}
$n$ & $x^{\rm cont}_n$ & rel.\ error & $x^{\rm cont}_n$ & rel.\ error
& $x^{\rm cont}_n$ & rel.\ error & $x^{\rm cont}_n$ & rel.\ error\\
\midrule
1 & 6.98793 & $+1.1\times10^{-9}$ & 2.34666 & $+1.0\times10^{-9}$ & 3.67824 & $+1.1\times10^{-9}$ & 4.89486 & $+1.1\times10^{-9}$\\
2 & 10.41712 & $-1.7\times10^{-9}$ & 7.92078 & $+7.8\times10^{-11}$ & 8.99277 & $-5.3\times10^{-10}$ & 10.11913 & $-9.4\times10^{-10}$\\
3 & 13.69802 & $-6.5\times10^{-9}$ & 11.38643 & $-3.3\times10^{-9}$ & 12.50540 & $-4.3\times10^{-9}$ & 13.68419 & $-5.1\times10^{-9}$\\
4 & 16.92362 & $-1.3\times10^{-8}$ & 14.69062 & $-8.7\times10^{-9}$ & 15.84021 & $-1.0\times10^{-8}$ & 17.05398 & $-1.1\times10^{-8}$\\
5 & 20.12181 & $-2.2\times10^{-8}$ & 17.93245 & $-1.6\times10^{-8}$ & 19.10373 & $-1.8\times10^{-8}$ & 20.34274 & $-1.9\times10^{-8}$\\
6 & 23.30425 & $-3.3\times10^{-8}$ & 21.14264 & $-2.5\times10^{-8}$ & 22.33014 & $-2.8\times10^{-8}$ & 23.58830 & $-2.9\times10^{-8}$\\
\bottomrule
\end{tabular}
\end{table}

Fitting the relative errors for $M=1280,\dots,20480$ to $M^{-q}$ gives
$q=1.998$--$1.999$ for all $p\leq8$ and $n\leq4$ ($q=2.000$ between $M=10240$ and
$20480$). Figure~\ref{fig:modes}(a) shows the case $p=0$: the error decreases by a
factor of four per doubling of $M$. In each angular sector the onion Laplacian is
thus a second-order discretization of the Laplacian on the cone, and the angular
cutoff and the truncation $j\geq p+1$ do not degrade the low sectors. Second order
and the weak $p$-dependence of the error coefficient reflect that the correction
in~\eqref{eq:Vexp} starts at order $\lambda^2$ with a relative size uniform in $p$
(Section~\ref{sec:candidates}). The convergence is not uniform in $p$, however: at
$M=20480$ the relative error of the lowest mode is $1.5\times10^{-8}$ even at
$p/M=0.75$, but $2\times10^{-3}$ for a chain of twenty layers ($p=20460$) and $0.41$
for a single layer ($p=20479$), and at fixed $p/M=\frac12$ the effective order drops
to $q\simeq1.3$--$1.4$. As for any second-order scheme, only a fixed fraction of the
spectrum is reproduced to a given accuracy (about $15\%$ to $10^{-2}$ at $p=0$, for
all $20\leq M\leq20480$), so observables dominated by the upper end of the radial
band reflect the lattice rather than the geometry.

\paragraph{Construction choices and the order of convergence}
The choices of Section~\ref{sec:onion} change the order of convergence even where
they leave the limit unchanged. With the completely positive map $U^{\rm pos}$
of~\eqref{eq:Upos} the limit is the same cone but the convergence drops to first
order ($6.4\times10^{-6}$ versus $1.0\times10^{-9}$ for $\Util$ in the lowest mode
of $p=1$ at $M=20480$), and endpoint weights $N_j^2$ or $N_{j+1}^2$ keep the
effective dimension $7$ but are likewise first order. The arithmetic midpoint
(Section~\ref{sec:KR}) keeps second order with an error coefficient about $13$ times
larger at $p=0$. Within the present construction, which writes the radial energy through
interlayer differences~\eqref{eq:EU} and adds the angular Laplacian~\eqref{eq:KL}
separately, the choice of $\Util$, i.e.\ of identifying modes by their
expansion coefficients as in~\cite{KT}, is thus singled out by second-order
convergence, not by the continuum limit.

\begin{figure}[t]
\centering
\includegraphics[width=\textwidth]{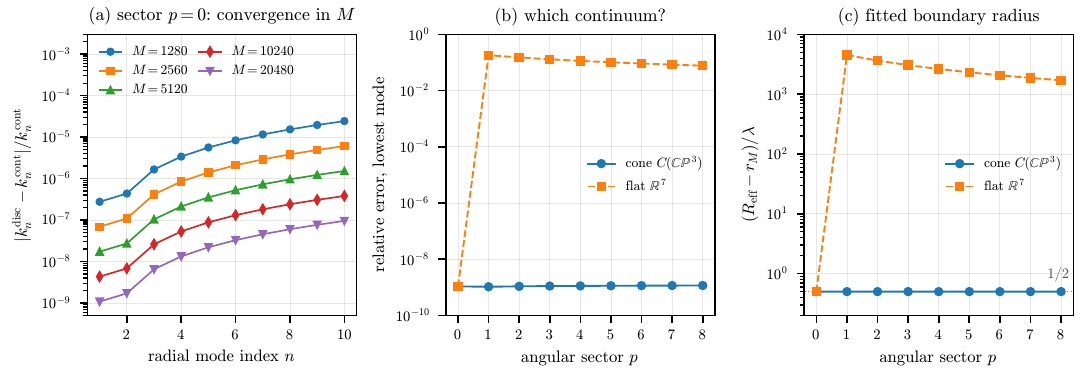}
\caption{Continuum limit of the fuzzy $\CP^3$ onion Laplacian.
(a) Relative deviation of the discrete radial wavenumbers $k^{\rm disc}_n$ from
the continuum values $k^{\rm cont}_n$ in sector $p=0$ ($M=1280,2560,5120,10240,20480$);
the error decreases by a factor of four per doubling of $M$.
(b) Relative error of the lowest radial mode at $M=20480$ under two continuum
hypotheses: the metric cone $C(\CP^3)$, with $\nu_p=\sqrt{(p+3/2)^2+4}$, and
flat $\mathbb{R}^7$, with $\nu=p+5/2$. The two coincide at $p=0$ and separate
for every $p\geq1$.
(c) Effective radius $R_{\rm eff}=x^{\rm cont}/k^{\rm disc}$ (lowest mode,
$M=20480$) when the outer boundary radius is treated as a free parameter. The cone
is fitted in all sectors by a single radius $R_{\rm eff}=r_M+\lambda/2$ (spread
$2.3\times10^{-6}\lambda$), whereas flat $\mathbb{R}^7$ requires a different
radius in each sector, with a spread of $4.5\times10^{3}\lambda$ ($22\%$ of the radius).
In (a) and (b) the outer radius is fixed at $R=r_M+\lambda/2$. To avoid roundoff
at large $M$, the eigenvalues are computed by shooting on the difference form of
the symmetrized chain.}
\label{fig:modes}
\end{figure}

\subsection{Discrimination from flat space and the outer radius}
\label{sec:discriminate}
\label{sec:robust}

The cone $C(\CP^3)$ and flat $\mathbb{R}^7$ share $\nu=5/2$ at $p=0$, so only the
sectors $p\geq1$ can discriminate between them. The only parameter in the
comparison is the outer radius $R$, which does not appear in the discrete
operator~\eqref{eq:tridiag} and enters only as the position of the continuum
boundary.

\paragraph{Fixed $R$}
For $R=r_M+\lambda/2$ the relative error of the lowest mode at $M=20480$ is
$(1.05$--$1.16)\times10^{-9}$ for the cone throughout $p=0,\dots,8$, and
$7.7\times10^{-2}$--$1.8\times10^{-1}$ for flat $\mathbb{R}^7$ at $p\geq1$
(Figure~\ref{fig:modes}b). The ratio reaches eight orders of magnitude with no free
parameter and grows as $M^2$. Shifting $R$ by half a cell, to $r_M$ or $r_{M+1}$,
raises the error relative to the cone to $\lambda/(2R)=2.4\times10^{-5}$ and makes
the convergence first order, but it remains more than three orders of magnitude
below the error relative to flat space.

\paragraph{$R$ determined from the data}
For each mode we determine the radius $R_{\rm eff}:=x^{\rm cont}_n/k^{\rm disc}_n$ at
which the discrete value and the continuum root agree (Figure~\ref{fig:modes}c); for
the correct candidate every mode in every sector should give the same $R_{\rm eff}$.
For the cone, all sectors $p=0,\dots,8$ give $(R_{\rm eff}-r_M)/\lambda=0.49998$ at
$M=20480$ (spread $2.3\times10^{-6}$), approaching $\frac12$ with deviations of
$O(1/M)$. For flat $\mathbb{R}^7$, $R_{\rm eff}$ is spread over $4.5\times10^{3}\lambda$
($22\%$ of the radius). The value $R=r_M+\lambda/2$ is thus selected by the spectrum,
not by us. (For the half-cell convention of Section~\ref{sec:KR} the same
procedure gives $R_{\rm eff}\to r_M$, with the same identification of the cone.)

\paragraph{Scale of the link}
\label{sec:linkscale}
The scale $\tilde g=4g_{\rm FS}$ follows from $r_j\simeq\lambda(j+1)$ with step
$\lambda$, i.e.\ from the Weyl rule of Section~\ref{sec:radius}; the matrix model
suggests a different ratio. The stacked configuration $X^A=\mu\bigoplus_k\rho_k(T^A)$
solves the Yang--Mills matrix model with a Myers term~\cite{Myers,Kitazawa}; its
layers lie in $\mathbb{R}^{15}$ at radii $\varrho_k=\mu\sqrt{3k(k+4)/8}$, a distance
$\mu\sqrt{3/8}$ apart for large $k$, with induced metric $\frac83\varrho_k^2g_{\rm FS}$.
String modes $|z\rangle\langle w|$~\cite{IKK,SteinackerString,SteinackerTekel}
between two layers satisfy $\langle\Box\rangle=|x(z)-x(w)|^2$ for
$\Box=\sum_A[X^A,[X^A,\cdot\,]]$ and $x(z)=\langle z|X|z\rangle$, up to the
zero-point term of~\cite{SteinackerString}, so this embedding
distance is the one the matrix model measures. Taking the radial step equal to the
layer spacing replaces $\hat C_2$ in $K_L$ by $\frac32\hat C_2$ and gives the cone
with link $\frac83g_{\rm FS}$ and $\nu_p^2=\frac{25}{4}+\frac32p(p+3)$. Since the
embedding is minimal, $p=1$ then also coincides with flat $\mathbb{R}^7$
($\nu_1=\frac72$), and the hypotheses separate only from $p=2$ on. The tests carry
over: at $M=20480$ the lowest mode agrees with this cone to
$(1.06$--$1.17)\times10^{-9}$ for $p\leq8$ with $q=2.00$, flat space misses by
$3.4\times10^{-2}$ at $p=2$ and $1.2\times10^{-1}$ at $p=8$, and $R_{\rm eff}$ is
again $0.49998$ for the cone but spread over $11\%$ of the radius for flat space.
The identification of a cone over $\CP^3$ rather than flat space is therefore
independent of the link scale; only the scale itself and the first discriminating
sector depend on it. The matrix model does not fix this scale, since it produces no
radial kinetic term, and we keep $\tilde g=4g_{\rm FS}$.

\subsection{The free propagator}
\label{sec:propagator}

The spectrum determines the free field theory defined by~\eqref{eq:action}. We take
$a=\frac12$, $b=\frac12m^2$ with $m>0$ and $c=0$. Here $m$ is the mass and $H^{(p)}$ is the symmetric tridiagonal matrix of~\eqref{eq:tridiag}. The mass is needed: for $m=0$ the
Neumann chain of the sector $p=0$ has the constant zero mode ($u_j\propto N_j$,
section~\ref{sec:rate}), so that $H^{(0)}$ is not invertible, and a massless
propagator would require projecting out this mode. We expand
$\Phi^{(j)}=\sum_{p,m}c^{(j)}_{pm}\hat Y^{(j)}_{pm}$ in a real basis normalized by the
integration rule~\eqref{eq:intrule}. In this subsection $k$ denotes a second layer index.
By~\eqref{eq:radint} the action becomes, with $c_j:=c^{(j)}_{pm}$,
$\frac12\VolCP\sum_{p,m}\sum_j\lambda r_j^6\,c_j\,[(-\Delta^{(p)}+m^2)c]_j$, the
discretization of $\frac12\VolCP\int r^6dr\,c\,(-L_p+m^2)c$, so that
$\langle c^{(j)}_{pm}c^{(k)}_{p'm'}\rangle=\delta_{pp'}\delta_{mm'}\,G^{(p)}_{jk}/\VolCP$. Here $G^{(p)}_{jk}$ is the free propagator between layers $j,k$ in angular sector $p$, and
\begin{equation}\label{eq:Gdisc}
G^{(p)}_{jk}=\frac{\bigl[(H^{(p)}+m^2)^{-1}\bigr]_{jk}}{\lambda\,r_j^3\,r_k^3}.
\end{equation}
The continuum
counterpart is the Green function of $-L_p+m^2$ with respect to $r^6dr$, regular at the
apex and satisfying the Neumann condition at $R$,
\begin{equation}\label{eq:Gcone}
\begin{gathered}
G_p(r,r')=(r_<r_>)^{-5/2}\,I_{\nu_p}(mr_<)\bigl[K_{\nu_p}(mr_>)+\kappa_p\,I_{\nu_p}(mr_>)\bigr],\\
\kappa_p=-\frac{xK_{\nu_p}'(x)-\frac52K_{\nu_p}(x)}{xI_{\nu_p}'(x)-\frac52I_{\nu_p}(x)}\bigg|_{x=mR},
\end{gathered}
\end{equation}
with $r_<=\min(r,r')$ and $r_>=\max(r,r')$, where $I_\nu$ and $K_\nu$ are the modified
Bessel functions of the first and second kind; its normalization follows from the Wronskian
of $I_\nu$ and $K_\nu$. Flat $\mathbb{R}^7$ corresponds, sector by sector, to
$\nu=p+\frac52$. In position space the addition theorem on $\CP^3$ gives
\begin{equation}\label{eq:Gpos}
\langle\phi(x)\phi(y)\rangle=\frac{1}{\VolCP}\sum_p d_p\,Z_p(\vartheta)\,G_p(r,r'),
\qquad
Z_p(\vartheta)=\frac{P^{(2,0)}_p(\cos\vartheta)}{P^{(2,0)}_p(1)},
\end{equation}
where $\vartheta$ is the geodesic distance on $\CP^3$ in the metric $\tilde g=4g_{\rm FS}$
$P^{(2,0)}_p$ is a Jacobi polynomial, and $Z_p$ is the zonal function of sector $p$
normalized to $Z_p(0)=1$. On the onion we represent $\Phi^{(j)}$ by the
function with the same expansion coefficients, the symmetric symbol of~\cite{DHMO}; the
two-point function is then~\eqref{eq:Gpos} with $G^{(p)}_{jk}$ in place of $G_p$ and the
sum truncated at $p\leq\min(j,k)-1$. We study convergence in this representation.
If the field is instead represented by its coherent-state symbol
$\langle z|\Phi^{(j)}|z\rangle$, the sector-$p$ term acquires the factor
$[T_{j-1}(p)\,T_{k-1}(p)]^{1/2}$ of~\eqref{eq:UposT}, where
$T_L(p)=1-p(p+3)/L+O(L^{-2})$. This is a first-order finite-layer correction of
relative size $\simeq\frac12p(p+3)\lambda(1/r+1/r')$ (the factor is $0.999512$ for
$p=1$ between the layers of Table~\ref{tab:prop}); it leaves the continuum limit
unchanged but reduces the convergence to first order.

The low-energy propagator of the onion is that of the cone. In Table~\ref{tab:prop} ($M=20480$) the
deviation from~\eqref{eq:Gcone} does not exceed $2.7\times10^{-8}$, whereas flat $\mathbb{R}^7$, which
agrees at $p=0$, differs by $21$--$37\%$ for $p=1,2,3$; in Figure~\ref{fig:prop}(a) ($M=20480$) the onion
points lie on the solid cone curves and are clearly separated from the dashed flat ones.
The deviation decreases as $M^{-2}$ in every sector (Figure~\ref{fig:prop}b; $q=2.00$
between $M=10240$ and $20480$, where it is below $2.2\times10^{-8}$ for $p\leq3$); already
at $M=40$ it is below $5\times10^{-3}$, far smaller than the difference from flat space; and
the position-space propagator~\eqref{eq:Gpos} of the symmetric symbols at
$\vartheta=0$, $\pi/2$ and $\pi$ converges at the same order from $M=160$ on, once the angular cutoff exceeds the
sectors that contribute. The layers are $j=\mathrm{round}(0.3M)$ and
$k=\mathrm{round}(0.6M)$, with the continuum functions evaluated at $r_j$ and $r_k$. The
agreement is a low-energy statement: at $M=20480$ the deviation in sector $p=1$ is
$3\times10^{-9}$ for $mR=5$ ($m\lambda=2.4\times10^{-4}$), but $2\times10^{-7}$ at
$m\lambda=10^{-3}$, $1.5\times10^{-5}$ at $4\times10^{-3}$, $1\times10^{-3}$ at $0.016$ and
$6\times10^{-2}$ at $0.06$. The propagator from the innermost layer,
finally, approaches the propagator from the apex, $\lim_{r\to0}G_0(r,r')$, with a
deviation of $1.4\times10^{-7}$ at $M=20480$ and an effective order that increases from
$1.7$ to $1.9$ over $M=80$--$20480$: the innermost layer takes the place of the apex also
for propagation.

\begin{table}[t]
\centering
\small
\caption{Sector propagators at $M=20480$, $\lambda=1$, $mR=5$, between the layers
$j=6144$ and $k=12288$ ($r\simeq0.3R$, $r'\simeq0.6R$), and on the layer $k$. Relative
deviations are taken with respect to the cone~\eqref{eq:Gcone} evaluated at $r_j$,
$r_k$; the last column compares the flat value $\nu=p+\frac52$ with the cone.}
\label{tab:prop}
\begin{tabular}{ccccc}
\toprule
$p$ & $R^5G_p(r_j,r_k)$ (cone) & onion, $(j,k)$ & onion, $(k,k)$ & flat $\mathbb{R}^7$ \\
\midrule
0 & 1.0779 & $-1.0\times10^{-9}$ & $-1.2\times10^{-8}$ & $0$ \\
1 & 0.6239 & $+3.2\times10^{-9}$ & $-1.6\times10^{-8}$ & $-0.211$ \\
2 & 0.3225 & $+1.0\times10^{-8}$ & $-2.1\times10^{-8}$ & $-0.313$ \\
3 & 0.1579 & $+2.2\times10^{-8}$ & $-2.7\times10^{-8}$ & $-0.368$ \\
\bottomrule
\end{tabular}
\end{table}

\begin{figure}[t]
\centering
\includegraphics[width=\textwidth]{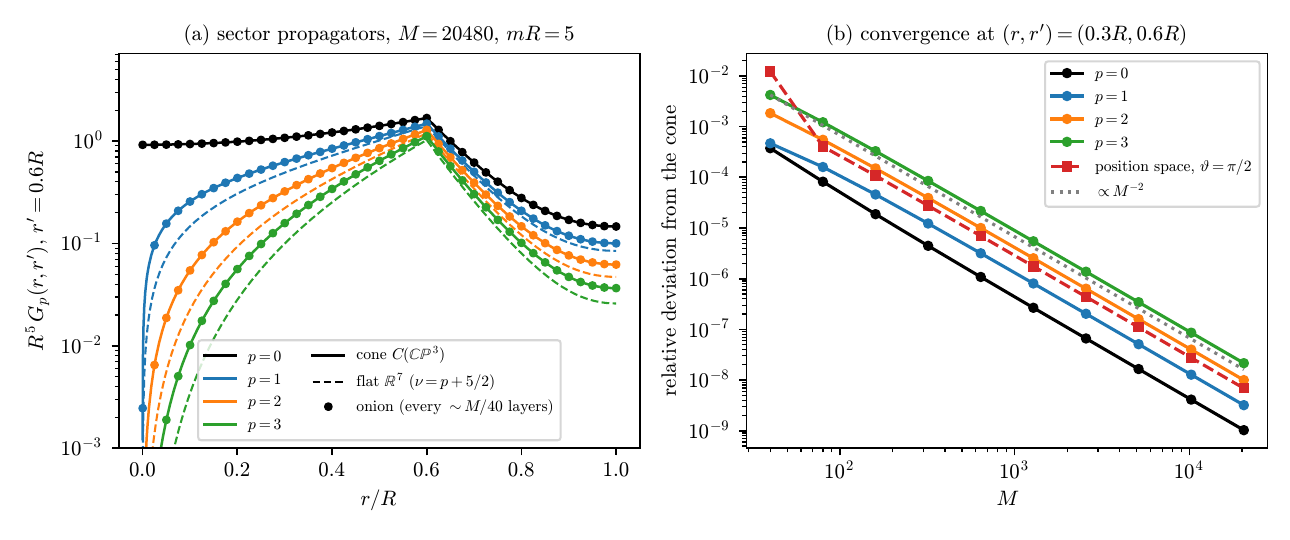}
\caption{Free propagator of the fuzzy $\CP^3$ onion ($\lambda=1$, $mR=5$).
(a) Sector propagators $G^{(p)}_{jk}$ from the layer $k$ at $r'\simeq0.6R$ for $M=20480$
(dots, shown every $\sim M/40$ layers), compared with the cone~\eqref{eq:Gcone} (solid) and with flat $\mathbb{R}^7$,
$\nu=p+\frac52$ (dashed; identical to the cone for $p=0$).
(b) Relative deviation from the cone at $(r,r')\simeq(0.3R,0.6R)$ for the sectors
$p=0,\dots,3$ and for the position-space propagator~\eqref{eq:Gpos} at $\vartheta=\pi/2$.
At $M=40$ the latter is dominated by the angular cutoff $p\leq j-1=11$. To avoid
roundoff at large $M$, the propagator is built from the two homogeneous solutions
propagated in difference form.}
\label{fig:prop}
\end{figure}

\FloatBarrier
\paragraph{Summary of this section}
\label{sec:contsummary}

Three conclusions follow. First, the fuzzy $\CP^3$ onion is seven-dimensional but
is not a discretization of flat seven-dimensional space: its continuum limit is the
cone over $\CP^3$, with link metric $4g_{\rm FS}$ for the radius convention of
Section~\ref{sec:radius}, and both the spectrum and the free propagator converge to
it at second order. Second, the cone is singular at its apex. A metric cone is
smooth there only if its link is the round unit sphere, and the cone over
$\CP^3$ with $\tilde g=4g_{\rm FS}$ has scalar curvature $-18/r^2$. The onion itself
never reaches the apex: its innermost layer is a single $1\times1$ block at
$r_1=6^{1/3}\lambda$, and the chain of sector $p$ begins at $r_{p+1}\simeq\lambda(p+2)$.
Nor does the free spectrum depend on how the apex is treated: because of the
centrifugal term $(\nu_p^2-\frac14)/r^2\geq6/r^2$ the endpoint $r=0$ is in the
limit-point case (the threshold being $\nu_p\geq1$), so the Laplacian of the cone is
self-adjoint without any boundary condition at the apex, and it is to its Bessel
spectrum that the onion converges. Third, the identification concerns the free
Laplacian and does not by itself imply that interacting theories on the onion have
$C(\CP^3)$ as their background.

\FloatBarrier
\section{Heat diffusion on the onion}
\label{sec:heat}

Section~\ref{sec:continuum} identified the geometry from the spectrum and the free
propagator. We now
examine the dynamics of classical fields on the onion. Heat diffusion adds no
information beyond the measure $h_j=N_j^2$ and the spectrum of $\Delta$, but it
expresses both in thermodynamic terms and, for localized initial conditions, makes
the angular cutoff directly visible (Section~\ref{sec:localized}).

\subsection{Heat equation and initial conditions}
\label{sec:heatsetup}

The heat equation (cf.~\cite{KT}) is
\begin{equation}\label{eq:heateq}
\Delta\Psi(t)=\alpha\,\partial_t\Psi(t),\qquad \Delta=K_R-K_L ,
\end{equation}
with $\alpha>0$ a constant setting the time scale.
By~\eqref{eq:KRpsi} and~\eqref{eq:KL}, it becomes an explicit evolution equation
for block-diagonal matrices,
\begin{align}\label{eq:heatexplicit}
\alpha\,\frac{\partial\Phi^{(j)}}{\partial t}
&=\frac{1}{\lambda^{2}}\left[
\frac{N_{j+1}}{N_j}D_j\Phi^{(j+1)}
+\frac{N_{j-1}}{N_j}\Util_{j-1}\Phi^{(j-1)}
-\frac{N_{j+1}+N_{j-1}}{N_j}\Phi^{(j)}\right]\nonumber\\
&\quad-\frac{1}{r_j^{2}}\sum_{A=1}^{15}\bigl[\rho(T^A),[\rho(T^A),\Phi^{(j)}]\bigr],
\end{align}
where the bracket describes radial heat conduction between neighbouring layers
and the last term angular diffusion within a layer. No basis appears. The
conserved heat is
\begin{equation}\label{eq:Qmat}
Q=\sum_{j=1}^{M}N_j\,\tr\,\Phi^{(j)}=\frac{1}{\mathcal V}\int d^7x\,\phi ,
\end{equation}
since $\tr K_L\Phi=0$ and, by self-adjointness and the zero mode,
$\langle\id,K_R\Psi\rangle_G=\langle K_R\id,\Psi\rangle_G=0$. In the $p=0$ sector
$\Phi^{(j)}=c_j(t)\,\id$, and we call $c_j$ the temperature of layer $j$; then
$Q=\sum_jN_j^2c_j$.

We take $M=12$, $\lambda=1$, $\alpha=1$ and the operator of Section~\ref{sec:KR}
(ghost term at the inner end of each chain, no outward bond at $j=M$). The initial
condition heats only the outermost layer, $c_{12}(0)=1$ and $c_j(0)=0$ otherwise.
To exhibit the angular sectors, the outermost amplitudes of the sectors
$p=1,2,3$ are also excited separately. The evolution is computed exactly by
eigendecomposition. The results are shown in Figure~\ref{fig:heat}.

\begin{figure}[t]
\centering
\includegraphics[width=\textwidth]{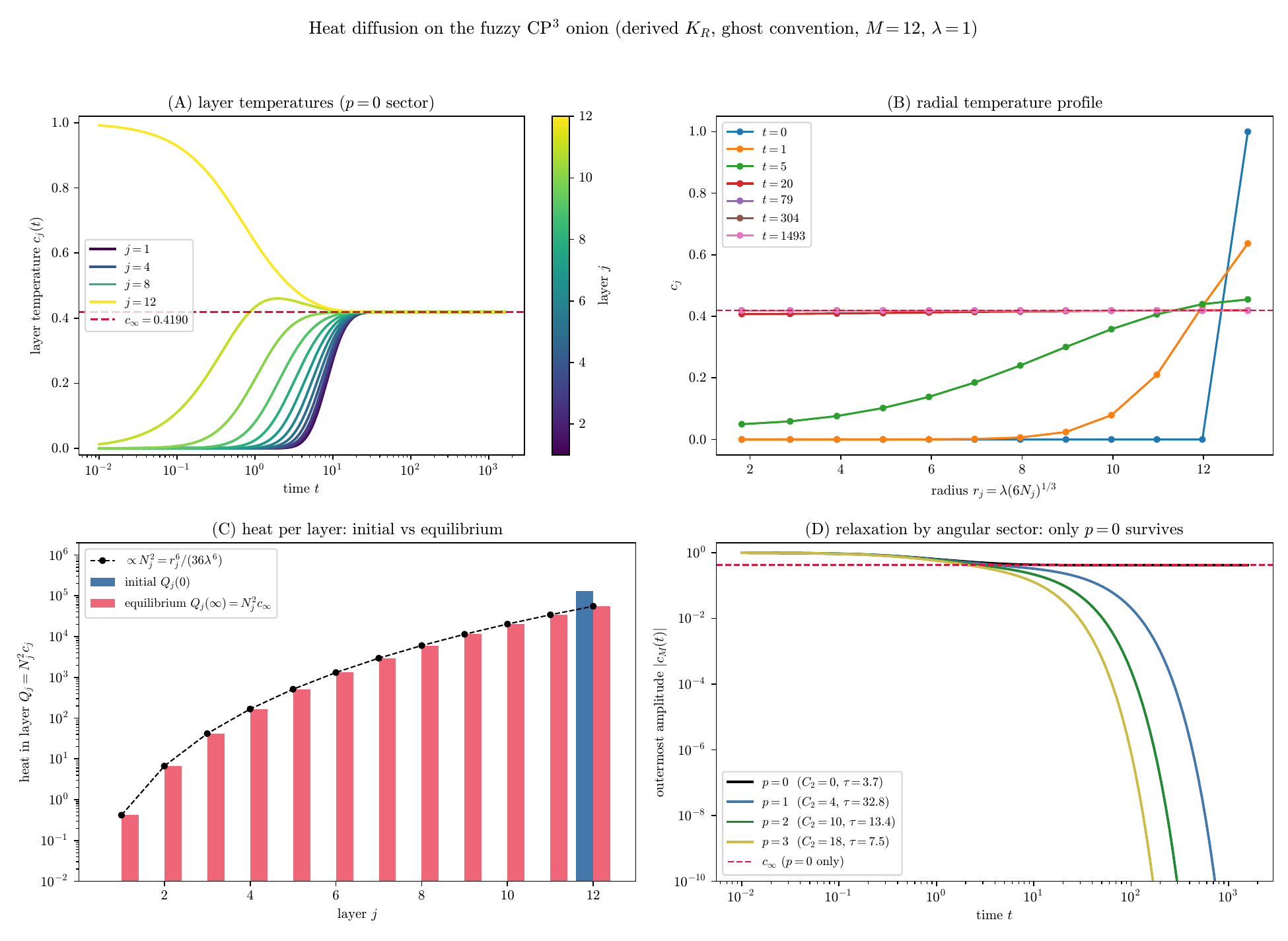}
\caption{Heat diffusion on the fuzzy $\CP^3$ onion ($M=12$, $\lambda=1$,
$\alpha=1$). (A) Temperature of each layer in sector $p=0$. (B) Radial
temperature profile. (C) Heat content of each layer, initially and at
equilibrium; the dashed line is $\propto N_j^2=r_j^6/(36\lambda^6)$.
(D) Relaxation of the outermost amplitudes in sectors $p=0,1,2,3$.}
\label{fig:heat}
\end{figure}

\subsection{Equilibrium plateau and relaxation times}
\label{sec:plateau}

All layers approach a common temperature (Figure~\ref{fig:heat}A), whose value
follows from the conservation law:
\begin{equation}\label{eq:cinf}
c_\infty=\frac{\sum_j N_j^2\,c_j(0)}{\sum_j N_j^2}
\;\overset{c_M(0)=1}{=}\;\frac{N_M^2}{\sum_{j=1}^{M}N_j^2}
=\frac{\dim\End(V_{M-1})}{\dim\{\Psi\}} .
\end{equation}
The equilibrium temperature equals the number of modes of the outermost layer
divided by the total number of degrees of freedom of the onion; for $M=12$,
$c_\infty=132496/316251=0.41896$. The plateau is fixed by the conservation law
and therefore by the measure alone; it restates the choice $h_j=N_j^2$ rather than
testing it.\footnote{With
the half-cell convention of Section~\ref{sec:KR}, $W_M=N_M^2/2$ for the outermost
layer, the total weight $\sum_jW_j$ changes by a
relative amount close to $D/(2M)$: $0.265$, $0.071$, $0.0181$, $0.0046$ for
$M=12,48,192,768$, compared with $7/(2M)=0.292$, $0.073$, $0.0182$, $0.0046$. The
two conventions thus differ by a boundary effect of order $\lambda/R$ rather than
a bulk effect of order $\lambda^2$. For $c_M(0)=1$ the deposited heat itself
differs, giving $c_\infty=0.4190$ (derived form) versus $0.2650$ (half cell), and
the slowest relaxation time of sector $p=0$ differs by $6.5\%$ at $M=12$.} At equilibrium the temperature is
uniform, but the heat content per layer, $Q_j=N_j\tr\Phi^{(j)}$, grows as
$Q_j\propto N_j^2\propto r_j^6$
(Figure~\ref{fig:heat}C), the volume of a seven-dimensional shell. The layer
$j=11$ adjacent to the source overshoots to $0.46$ at $t\simeq2$ before relaxing,
as expected for local diffusion.

Since $N_j\simeq j^3/6$, one has $\sum_jN_j^2\simeq M^7/252$ and
$N_M^2\simeq M^6/36$, hence, with $D$ the spatial dimension,
\begin{equation}\label{eq:Mcinf}
M\,c_\infty\;\longrightarrow\;7=D .
\end{equation}
\begin{table}[htbp]
\centering
\small
\caption{Equilibrium plateau multiplied by the number of layers, $Mc_\infty$, for the initial condition $c_M(0)=1$.}
\label{tab:Mcinf}
\begin{tabular}{rccccccc}
\toprule
$M$ & 320 & 640 & 1280 & 2560 & 5120 & 10240 & 20480 \\
\midrule
$Mc_\infty$ & 6.903 & 6.951 & 6.975 & 6.988 & 6.994 & 6.997 & 6.998 \\
\bottomrule
\end{tabular}
\end{table}
As Table~\ref{tab:Mcinf} shows, the deficit $7-Mc_\infty$ halves with each doubling of $M$, so that
$Mc_\infty=7-O(1/M)$. Equivalently, the equilibrium temperature per unit heat
scales as $R^{-D}$, with local exponent $-6.9998$ between $M=10240$ and $M=20480$. The
plateau thus expresses the $r^6$ growth of the measure as a thermodynamic
observable.

\paragraph{Relaxation times}
\label{sec:relax}

The slowest decay rate in sector $p$ is the smallest non-zero eigenvalue of
$-\Delta^{(p)}$ (with $\alpha=1$), and its inverse is the relaxation time $\tau_p$;
this restates the spectrum of Section~\ref{sec:continuum} at small $M$. The
cone predicts the decay rate to be $(x_1/R)^2$, with $x_1$ the first root
of~\eqref{eq:neumann} for $\nu=\nu_p$; Table~\ref{tab:relax} compares the two.
\begin{table}[htbp]
\centering
\small
\caption{Slowest relaxation times $\tau_p$ of the angular sectors for $M=12$, compared with the cone prediction $(R/x_1)^2$ with $R=r_{12}+\lambda/2=13.474$.}
\label{tab:relax}
\begin{tabular}{cccccc}
\toprule
$p$ & $C_2(p)$ & $x_1$ & $\tau_p$ (onion, $M=12$) & $(R/x_1)^2$ (cone) & rel.\ difference \\
\midrule
0 & 0 & 6.9879 & 3.701 & 3.718 & $-4.5\times10^{-3}$ \\
1 & 4 & 2.3467 & 32.812 & 32.969 & $-4.8\times10^{-3}$ \\
2 & 10 & 3.6782 & 13.352 & 13.419 & $-5.0\times10^{-3}$ \\
3 & 18 & 4.8949 & 7.538 & 7.578 & $-5.2\times10^{-3}$ \\
\bottomrule
\end{tabular}
\end{table}
Even for an onion of twelve layers, the cone
predicts the relaxation times to $0.5\%$. The ordering is not monotonic in $p$:
the dipolar sector $p=1$ relaxes about nine times more slowly than the lowest
non-zero mode of $p=0$, because for $\nu_1=3.2016$ the Neumann condition has a
small first root. The angular patterns with $p>0$ decay completely, leaving only
the uniform component (Figure~\ref{fig:heat}D).

\subsection{Localized initial conditions and the angular cutoff}
\label{sec:localized}

The initial conditions above excite a single sector. Figure~2 of~\cite{KT}
instead excites a single matrix element of the outermost layer, i.e.\ a hot spot
near the north pole. Its $\CP^3$ analogue is the projector onto a coherent
state~\eqref{eq:zstate} (with $k=j-1$ on layer $j$), localized at a point. Since all points of $\CP^3$ are equivalent, we take $z_0=(1,0,0,0)$ in
place of the north pole, call it the pole, and solve
\begin{equation}\label{eq:pole}
\Phi^{(M)}(0)=|z_0\rangle\langle z_0|,\qquad
|z_0\rangle=\frac{(a_1^\dagger)^{M-1}}{\sqrt{(M-1)!}}|0\rangle,
\qquad \Phi^{(j)}(0)=0\ \ (j<M),
\end{equation}
with the parameters of~\cite{KT}, $M=5$, $\lambda=1$, $\alpha=1$, exactly by
eigendecomposition. The initial condition is invariant under the transformations
that fix $z_0$, i.e.\ phase rotations of $a_1$ and unitary transformations of
$(a_2,a_3,a_4)$; since $\Delta$ commutes with $SU(4)$, so is $\Phi^{(j)}(t)$ at all
times. Hence $\Phi^{(j)}(t)$ is diagonal in the occupation-number basis
$|n_1,\dots,n_4\rangle$ ($n_i$ the occupation number of $a_i$), with diagonal entries
$\Phi^{(j)}_{n_1}(t)$ depending on $n_1$ only.

A single matrix element contains all sectors (Table~\ref{tab:polesectors}).
\begin{table}[htbp]
\centering
\small
\caption{Sector content of the localized initial condition~\eqref{eq:pole} for $M=5$: dimension $d_p$, layers on which sector $p$ exists, and fraction of $\|\Phi^{(M)}\|^2$.}
\label{tab:polesectors}
\begin{tabular}{cccc}
\toprule
$p$ & $d_p$ & Layers on which it exists & Fraction of $\|\Phi^{(M)}\|^2$ \\
\midrule
0 & 1 & $j=1,\dots,5$ & $2.86\%$ \\
1 & 15 & $j=2,\dots,5$ & $21.43\%$ \\
2 & 84 & $j=3,\dots,5$ & $40.00\%$ \\
3 & 300 & $j=4,5$ & $28.57\%$ \\
4 & 825 & $j=5$ & $7.14\%$ \\
\bottomrule
\end{tabular}
\end{table}
Since layer $j$ accommodates only $p\leq j-1$, $97.1\%$ of the hot spot consists
of modes that can never reach the innermost layer; only the $p=0$ component,
$2.86\%=1/N_M$, can. In the single-sector computation of
Section~\ref{sec:plateau} the cutoff never acts ($p\leq3\ll M=12$), whereas here
the finite matrix size shapes the classical diffusion problem directly. Since
$\tr\,Y^{(j)}_{pm}=0$ for $p>0$, only the $p=0$ component carries heat, and
\begin{equation}
Q=N_M=35,\qquad \sum_{j=1}^{5}N_j^2=1742,\qquad c_\infty=\frac{Q}{\sum_jN_j^2}=0.0200918 .
\end{equation}
The innermost layer, a $1\times1$ matrix, approaches $c_\infty$.

The value at the pole, $\langle z_0|\Phi^{(j)}(t)|z_0\rangle=\Phi^{(j)}_{j-1}(t)$, at the
times used in~\cite{KT} is listed in Table~\ref{tab:poletemp}.
\begin{table}[htbp]
\centering
\small
\caption{Value at the pole, $\langle z_0|\Phi^{(j)}(t)|z_0\rangle$, of each layer for the localized initial condition~\eqref{eq:pole} ($M=5$, $\lambda=\alpha=1$), at the times used in~\cite{KT}.}
\label{tab:poletemp}
\begin{tabular}{crrrrr}
\toprule
$t$ & $j=1$ & $j=2$ & $j=3$ & $j=4$ & $j=5$ \\
\midrule
0 & 0 & 0 & 0 & 0 & 1.000000 \\
0.2 & 0.000042 & 0.001245 & 0.019656 & 0.180213 & 0.847999 \\
0.4 & 0.000425 & 0.006603 & 0.054248 & 0.266876 & 0.739318 \\
0.6 & 0.001397 & 0.015071 & 0.086272 & 0.305144 & 0.657277 \\
0.8 & 0.002919 & 0.024631 & 0.110907 & 0.318111 & 0.592502 \\
3 & 0.017983 & 0.062981 & 0.130482 & 0.206357 & 0.270355 \\
30 & 0.020092 & 0.021258 & 0.021960 & 0.022426 & 0.022721 \\
\bottomrule
\end{tabular}
\end{table}
The times follow Figure~2 of~\cite{KT}. At $t=0.8$ the heat is still
concentrated in the outer layers, and equilibrium is nearly reached only at $t\sim30$
(Figure~\ref{fig:heatloc}).

\begin{figure}[tbp]
\centering
\includegraphics[width=\textwidth]{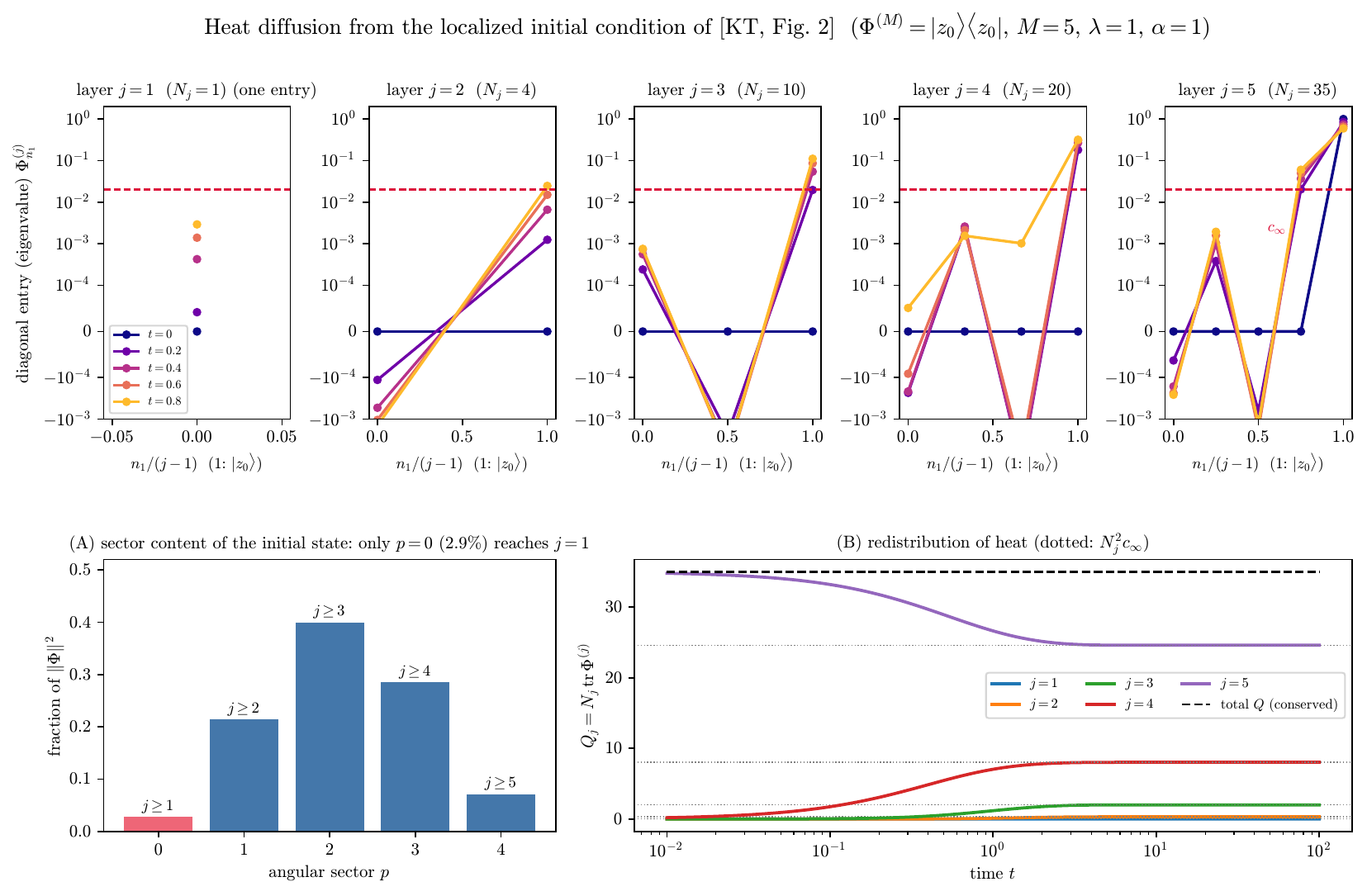}
\caption{Heat diffusion from the localized initial condition~\eqref{eq:pole}
($M=5$, $\lambda=1$, $\alpha=1$). Top: diagonal entries (${}=$ eigenvalues)
$\Phi^{(j)}_{n_1}$ of $\Phi^{(j)}(t)$ in each layer as functions of $n_1/(j-1)$. The
horizontal axis labels eigenstates, not positions; $n_1=j-1$ is the coherent state
$|z_0\rangle$ at the pole. (A) Sector content of
the initial state. (B) Heat content of each layer; the total is conserved and each
layer approaches $N_j^2c_\infty$ (dotted).}
\label{fig:heatloc}
\end{figure}

\begin{figure}[htbp]
\centering
\includegraphics[width=\textwidth]{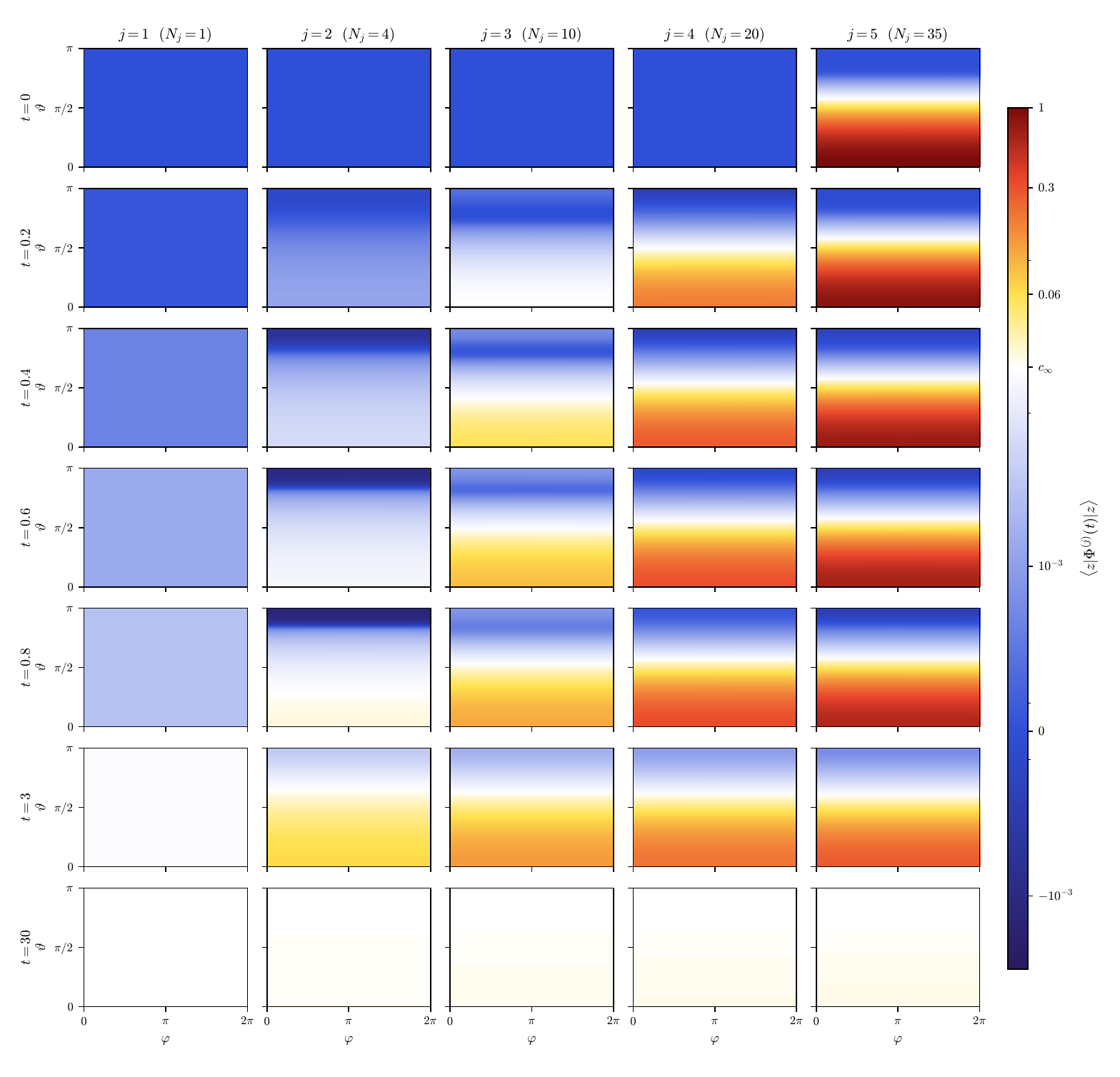}
\caption{$f^{(j)}=\langle z|\Phi^{(j)}|z\rangle$ of each layer for the localized initial
condition~\eqref{eq:pole} ($M=5$, $\lambda=1$, $\alpha=1$) on the $\CP^1$ through the
pole (a unit sphere), in the layout of Figure~2 of~\cite{KT}: columns are the
layers $j=1,\dots,5$, rows the times $t=0,\,0.2,\,0.4,\,0.6,\,0.8,\,3,\,30$. The
colour scale is symmetric-logarithmic, and white marks the equilibrium value
$c_\infty=0.0201$. $f^{(j)}$ is not a pointwise
temperature but a smoothing of the diagonal entries in the top row of
Figure~\ref{fig:heatloc}, and area on the sphere is not proportional to volume in
$\CP^3$ (the northern half is $1/8$ of the volume; $\vartheta=\pi$ is a whole
$\CP^2$). The dark blue negative values come from negative eigenvalues of
$\Phi^{(j)}$ and do not represent negative temperatures; see the text.}
\label{fig:heatsnap}
\end{figure}

To view the matrices as functions on $\CP^3$ we use
$f^{(j)}(z,t):=\langle z|\Phi^{(j)}(t)|z\rangle$; $f^{(j)}(z_0,t)$ is the value of
Table~\ref{tab:poletemp}. A positive matrix gives a non-negative function, but
$\Phi^{(j)}(t)$ need not be positive (see the end of this section), so $f^{(j)}$ can
take negative values. By the invariance above, $f^{(j)}$ depends on $|z_1|^2$
only, and $\vartheta$ defined by $|z_1|^2=\cos^2(\vartheta/2)$ is the geodesic
distance from $z_0$ in the metric $\tilde g=4g_{\rm FS}$ of Section~\ref{sec:trace}.
The $\CP^1$ through $z_0$, $z=(\cos\frac\vartheta2,\,e^{i\varphi}\sin\frac\vartheta2,\,0,\,0)$,
is a round sphere of radius one in this metric, so $f^{(j)}$ restricted to it is an
axially symmetric function on the sphere, as in Figure~2 of~\cite{KT}.
Figure~\ref{fig:heatsnap} shows it in the layout of that figure.

Two points must be kept in mind when reading Figure~\ref{fig:heatsnap}. First,
$f^{(j)}$ is not the value of the temperature at the point $z$ but the
coherent-state expectation value (Berezin symbol) of $\Phi^{(j)}$~\cite{Perelomov}.
Since $\Phi^{(j)}$ is diagonal and depends on $n_1$ only,
\[
f^{(j)}(\vartheta,t)=\sum_{n_1=0}^{j-1}\binom{j-1}{n_1}x^{n_1}(1-x)^{j-1-n_1}\,\Phi^{(j)}_{n_1}(t),
\qquad x=\cos^2\frac\vartheta2 .
\]
Thus $f^{(j)}$ is the Bernstein polynomial of the diagonal entries
$\Phi^{(j)}_{n_1}$ in the top row of Figure~\ref{fig:heatloc}, i.e.\ their smoothing
over a width $\sqrt{x(1-x)/(j-1)}$, and layer $j$
shows no angular structure finer than $\sim1/\sqrt{j-1}$. The width of the
outermost peak $f^{(5)}=\cos^8(\vartheta/2)$ at $t=0$ is that of the coherent state
$|z_0\rangle$ itself, and $f^{(1)}$ of the $1\times1$ innermost layer is constant.
Only the pole value is not smoothed: $f^{(j)}(z_0,t)=\Phi^{(j)}_{j-1}(t)$. Second,
the figure is a two-dimensional slice of the six-dimensional $\CP^3$, and area on
the sphere does not represent volume in $\CP^3$. At distance $\vartheta$ from $z_0$
the volume element of $\CP^3$ is $\propto\sin^5(\vartheta/2)\cos(\vartheta/2)\,d\vartheta$~\cite{BZ}
(on $S^2$, $\propto\sin\vartheta\,d\vartheta$), and $\vartheta=\pi$ is not a point but the
whole $\CP^2$ at $z_1=0$. The northern half $\vartheta<\pi/2$, which is half the area
of $S^2$, is only $1/8$ of the volume of $\CP^3$. The same value therefore stands for
more heat the further south it is, and heat that appears to stay near the pole in
the figure need not stay there in amount. The heat of a layer is
$Q_j=N_j\tr\Phi^{(j)}=N_j^2\int f^{(j)}d\mu$, with $d\mu$ the volume element
normalized to total volume one; it is shown in Figure~\ref{fig:heatloc}(B).

The diagonal entries in the top row of Figure~\ref{fig:heatloc} show small
negative values. Since $\Phi^{(j)}$ is diagonal, these are negative eigenvalues, and the dark blue near $\vartheta=\pi$ in
Figure~\ref{fig:heatsnap} (at most $-1.4\times10^{-3}$, at $t=0.8$) is their
smoothing: the evolution~\eqref{eq:heatexplicit} does not preserve positivity.
$\Util$ is not a positive map: transporting the
pole projector of layer $j=2,\dots,6$ to the next layer with $\Util$ produces a
negative eigenvalue of $-6\%$ to $-1.5\%$ of the largest one ($U_{2q}$ preserves
positivity). We measure the most negative eigenvalue during the evolution in two
ways: relative to the largest eigenvalue over all layers and times (global ratio),
and relative to the largest eigenvalue of the same layer (layer ratio). A positive
interlayer map alone does not remove the negative eigenvalues. With the
completely positive map $U^{\rm pos}$ of~\eqref{eq:Upos} in place of $\Util$, the
largest layer ratio at $M=5$ drops from about $13\%$ to about $2\%$ but does not vanish.
The negative values are, on the other hand, a finite-size effect and decrease as $M$ grows at
fixed outer radius. For the pole projector~\eqref{eq:pole} the global ratio
falls roughly as $1/M$, from
$1.7\times10^{-2}$ at $M=5$ to $9.4\times10^{-4}$ at $M=64$, while the layer ratio in
the outer layers ($r\geq\frac34R$) falls more slowly, from $7\%$ at $M=8$ to $3\%$ at
$M=64$. The negative values of the Berezin symbol in Figure~\ref{fig:heatsnap}, and
the negative eigenvalues for an initial condition spread to a finite angular width
by $e^{-\sigma\hat C_2}$, vanish exponentially in $M$.
The heat equation~\eqref{eq:heateq} is therefore to be read as the linear
relaxation generated by $\Delta$, which probes its spectrum and measure, rather
than as heat conduction of a positive density. Only the $p=0$ part of $\Phi^{(j)}$,
proportional to the identity ($c_j$ of Section~\ref{sec:heatsetup}), has the meaning
of a temperature, and the negative values in Figures~\ref{fig:heatloc}
and~\ref{fig:heatsnap} do not mean negative temperatures.

\FloatBarrier
\section{The twistor fibre}
\label{sec:kk}

$\CP^3$ is an $S^2$ bundle over $S^4$ preserved by $Sp(2)\cong Spin(5)\subset SU(4)$,
i.e.\ the twistor space of $S^4$~\cite{MHOD,Steinacker}. In this section we show
that this fibre structure does not appear in the spectrum of the onion as
constructed, and that among the $Spin(5)$-invariant deformations that make the fibre
visible, those obtainable by rescaling the $\mathfrak{su}(4)$ generators open no
five-dimensional energy window.

\label{sec:fibre}

Under $Spin(5)$ each angular sector decomposes as
\begin{equation}\label{eq:branch}
(p,0,p)\;\longrightarrow\;\bigoplus_{s=0}^{p}\,(p-s,2s),\qquad
\dim(a,b)=\frac{(a+1)(b+1)(a+b+2)(2a+b+3)}{6}
\end{equation}
\cite{Steinacker}, where $(a,b)$ are Dynkin labels of $Spin(5)$; the
components with $s=0$ are harmonics on $S^4$, and those with $s\geq1$ transform as
spin $s$ on $S^4$. The operator $\Delta$ commutes with $SU(4)$, and $(p,0,p)$ occurs
exactly once in each $\End(V_{j-1})$ with $j-1\geq p$. By Schur's lemma, the restriction of $\Delta$
to the $(p,0,p)$-isotypic component of the whole onion is
$\id_{d_p}\otimes\Delta^{(p)}$, with $\Delta^{(p)}$ the chain operator of
Section~\ref{sec:reduction}. All components of~\eqref{eq:branch} are therefore
exactly degenerate, and the Kaluza--Klein mass gap associated with the fibre
vanishes identically. Seen from $S^4$, the fibre modes behave as a tower of
massless higher-spin fields, consistent with the interpretation
of~\cite{Steinacker}. Separating them requires a term that breaks $SU(4)$, i.e.\ a
squashing of the twistor bundle.

\subsection{Squashing and matrix configurations}
\label{sec:kk-setup}

We decompose $\mathfrak{su}(4)=\mathfrak{sp}(2)\oplus\mathfrak m$ orthogonally, where
$\mathfrak{sp}(2)=\{X\in\mathfrak{su}(4):X^T\Omega+\Omega X=0\}$ and
$\Omega=\mathrm{diag}(i\sigma_2,i\sigma_2)$ ($\sigma_2$ a Pauli matrix). With
bases $S^a$ ($a=1,\dots,10$) of $\mathfrak{sp}(2)$ and $P^\alpha$ ($\alpha=1,\dots,5$) of
$\mathfrak m$, normalized as $4\times4$ matrices by
$\tr(S^aS^b)=\tr(P^\alpha P^\beta)=\frac12\delta$ (and $\tr(S^aP^\alpha)=0$ since
$\mathfrak{sp}(2)\perp\mathfrak m$), we define on each layer
\begin{equation}\label{eq:C5}
\hat C_2^{Spin(5)}\Phi^{(j)}=\sum_{a=1}^{10}\bigl[\rho(S^a),[\rho(S^a),\Phi^{(j)}]\bigr],\qquad
\hat\Sigma_5\Phi^{(j)}=\sum_{\alpha=1}^{5}\bigl[\rho(P^\alpha),[\rho(P^\alpha),\Phi^{(j)}]\bigr].
\end{equation}
Then $\hat C_2=\hat C_2^{Spin(5)}+\hat\Sigma_5$, and on $(p-s,2s)$ one has
$\hat C_2^{Spin(5)}=\frac12[p(p+3)+s(s+1)]$ and
$\hat\Sigma_5=\frac12[p(p+3)-s(s+1)]$ (\cite{Steinacker},
where the Casimirs are normalized to twice ours). The fibre Laplacian
of~\cite{MHOD},
\begin{equation}\label{eq:DeltaI}
\Delta_I:=2\hat C_2^{Spin(5)}-\hat C_2,\qquad
\Delta_I\big|_{(p-s,2s)}=s(s+1),
\end{equation}
therefore has the spectrum of a round sphere of unit radius, independent of both
$p$ and the layer.

The squashing of~\cite{MHOD} replaces $\hat C_2$ in the angular
Laplacian~\eqref{eq:KL} by $\hat C_2+h\,\Delta_I$. The value $h=0$ is the round
$\CP^3$, and the ratio of the radii $R_{S^2}$ and $R_{S^4}$ of the fibre $S^2$ and the
base $S^4$ is
\begin{equation}\label{eq:Rfibre}
\frac{R_{S^2}}{R_{S^4}}=\frac{1}{\sqrt{1+h}}
\end{equation}
\cite{MHOD}. The range in which this input can arise from rescaled
configurations of the generators is determined by the identity
\begin{equation}\label{eq:KLsplit}
\hat C_2+h\,\Delta_I=(1+h)\,\hat C_2^{Spin(5)}+(1-h)\,\hat\Sigma_5 .
\end{equation}
Since $\mathfrak{sp}(2)$ and $\mathfrak m$ are inequivalent irreducible
representations of $Sp(2)$, Schur's lemma restricts the $Sp(2)$-invariant
double-commutator Laplacians of configurations linear in the generators to
$a\,\hat C_2^{Spin(5)}+b\,\hat\Sigma_5$ with $a,b\geq0$. The squashing lies on the
line $(a,b)=(1+h,1-h)$ (Figure~\ref{fig:squash}(a)), so for $-1\leq h\leq1$ it arises
from the rescaled configuration $Y$: $\rho(S^a)\to\sqrt{1+h}\,\rho(S^a)$,
$\rho(P^\alpha)\to\sqrt{1-h}\,\rho(P^\alpha)$, while outside this range it arises from
no configuration linear in the generators. Since
$\Box\,\rho(S^a)=(4+2h)\,\rho(S^a)$ and $\Box\,\rho(P^\alpha)=4\,\rho(P^\alpha)$ with
$\Box=\sum_B[Y^B,[Y^B,\cdot\,]]$ ($Y^B$ the fifteen rescaled matrices), $Y$ is an exact classical solution of the
Yang--Mills matrix model with separate $Sp(2)$-invariant (negative) mass terms for the
two sets of matrices (cf.\ the $SO(5)$-invariant effective model of~\cite{SperlingSteinacker},
which has separate, possibly negative, mass terms for two sets of five matrices).

The restriction to rescaled generators is not essential. The Laplacian of an
arbitrary configuration of Hermitian matrices may be averaged over $Sp(2)$ without
changing the target, and since $\End(V_{j-1})$ is multiplicity-free also under $Sp(2)$
(the label $(p-s,2s)$ in~\eqref{eq:branch} determines $(p,s)$), the average is fixed by
one non-negative weight per $Sp(2)$-irreducible component of the configuration.
Matching the spectrum $p(p+3)+h\,s(s+1)$ is linear in these weights; for
$k=j-1\leq6$ the unique solution is supported on the rescaled generators and is
non-negative exactly for $-1\leq h\leq1$ (the obstruction lies in the sectors $p=2$,
figure~\ref{fig:squash}(b); the coefficient matrix and the procedure are given in
appendix~\ref{app:lp}). Approximately the situation differs: the relative error
$\varepsilon$ needed to match the sectors $p\leq2$ decreases with $k$
(Figure~\ref{fig:squash}(c)), the optimal configurations consisting of operators along
the fibre (products of the $\mathfrak{sp}(2)$ generators). In the continuum, the
horizontal distribution of the twistor fibration (the tangent vectors orthogonal to
the fibre) generates the whole tangent space under brackets, so no Hamiltonian
vector field points purely along the fibre; but functions oscillating rapidly along
the fibre have Hamiltonian vector fields nearly along it, which suggests that any $h>-1$ can be approached as
$k\to\infty$. In~\cite{MHOD} the term $h\Delta_I$ with large $h$ is added to the
field action as a penalty rather than obtained from a configuration.

\begin{figure}[t]
\centering
\includegraphics[width=\textwidth]{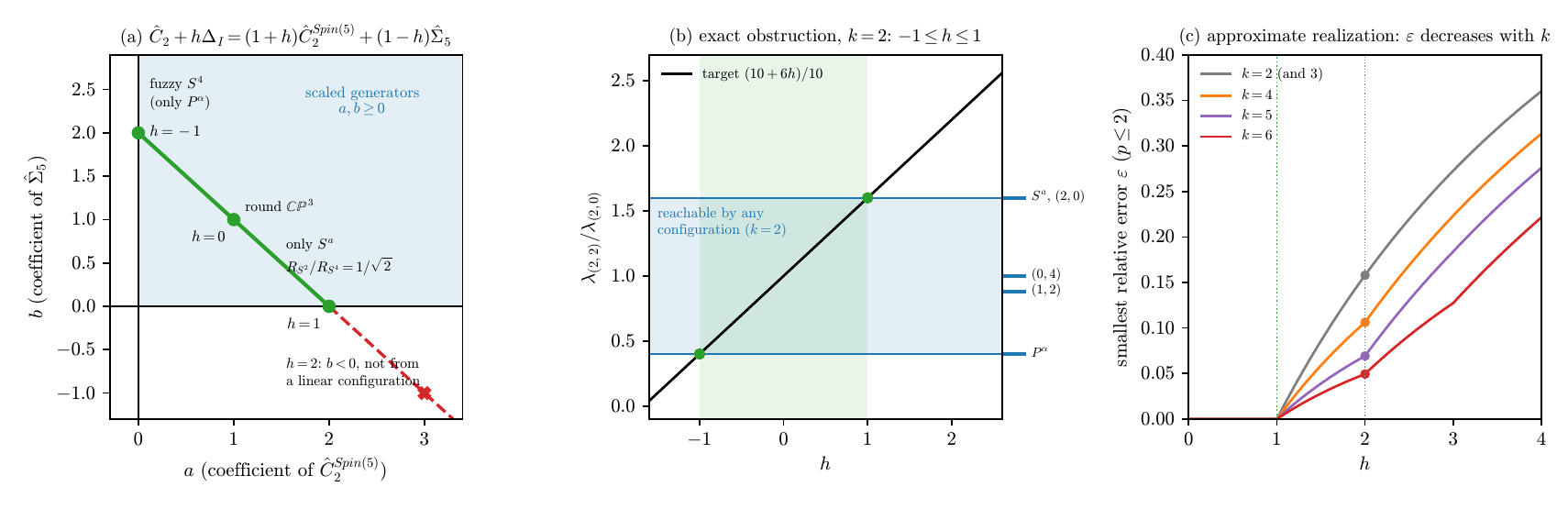}
\caption{Realizing the squashing $\hat C_2+h\Delta_I$ by matrix configurations.
(a) The $Sp(2)$-invariant double-commutator Laplacians of configurations linear in the
generators are $a\,\hat C_2^{Spin(5)}+b\,\hat\Sigma_5$ with $a,b\geq0$ (shaded). The
squashing is the line $(a,b)=(1+h,1-h)$ (identity~\eqref{eq:KLsplit}), which leaves the
shaded region for $h>1$. The endpoint $h=-1$ is the Laplacian of fuzzy $S^4$, and $h=1$ is
the configuration of the $\rho(S^a)$ alone.
(b) Exact obstruction at $k=2$. Let $\lambda_{(p,s)}$ be the eigenvalue on the sector
$(p-s,2s)$ of the $Sp(2)$-averaged Laplacian of a configuration. For an arbitrary
configuration of Hermitian matrices the ratio $\lambda_{(2,2)}/\lambda_{(2,0)}$ lies in
$[0.4,\,1.6]$ (values for each $Sp(2)$-irreducible component are marked on the right;
derivation in Appendix~\ref{app:lp}). The target $(10+6h)/10$ (solid line) lies in this
range exactly for $-1\leq h\leq1$.
(c) Approximate realization: the smallest relative error $\varepsilon$ with which the
spectrum in the sectors $p\leq2$ can be reproduced (linear programming). It vanishes
for $-1\leq h\leq1$ and decreases with $k$ for $h>1$ ($0.16$, $0.11$, $0.069$, $0.049$ at
$h=2$ for $k=3,4,5,6$; $k=3$ coincides with $k=2$).}
\label{fig:squash}
\end{figure}

\subsection{No five-dimensional window}
\label{sec:kk-interp}

The squashed onion is
block-diagonal in $(p,s)$, and each block is the chain~\eqref{eq:tridiag} with the
centrifugal term replaced by $[p(p+3)+h\,s(s+1)]/r_j^2$. In terms of the energy $E$ (an eigenvalue of
$-\Delta$), a fibre level $s$ enters the density of states at the Kaluza--Klein scale of the fibre at the outer
boundary, $E\simeq s(s+1)(1+h)/R^2$, so that a five-dimensional regime can exist
only in the energy range
\begin{equation}\label{eq:range5d}
\frac{1}{R^2}\ll E\ll\min\Bigl(\frac{2(1+h)}{R^2},\,\frac{1}{\lambda^2}\Bigr).
\end{equation}
For $h\leq1$, the range realized exactly by matrix configurations (at least up to
$k=6$), the upper limit of this range is at most $4/R^2$, at most a factor of four above the lower limit; no
energy satisfies the condition, and no intermediate five-dimensional regime exists. The exact modes
confirm this: with the Neumann condition~\eqref{eq:neumann}, the energy of the lowest mode
with a fibre excitation, $(p,s)=(1,1)$, is only $1.49$ times that of the lowest mode of
$(1,0)$ at $h=1$ (and $5.7$ times at $h=10$), in the continuum cone and in the discrete chains alike
($M=80$--$20480$), close to the ratio $(4+2h)/4$ of the angular eigenvalues. The fibre radius at
the outer boundary also stays at $R/\sqrt{1+h}\geq R/\sqrt2$, far larger than the
noncommutativity scale $\lambda$. Such a range appears when a term with
$h\gg1$ is added to the action, as in~\cite{MHOD}, or possibly when it is
approximated by configurations of high order at large $k$. The density of states
then crosses over from five to seven dimensions according to Weyl's law, but since
this depends on the externally given value of $h$ and does not change the conclusion
that the radial operator is that of the cone (next paragraph), we do not pursue it
here.

Even if $h\to\infty$ is imposed from outside, the onion does not become flat
$\mathbb{R}^5$. The $s=0$ chains do not depend on $h$, so their continuum limit
remains the operator~\eqref{eq:Lp} with $\nu_p=\sqrt{(p+\frac32)^2+4}$, not the flat
cone $C(S^4)=\mathbb{R}^5$ with $\nu=p+\frac32$: the radial measure
$N_j^2\propto r^6$ inherited from $\CP^3$ is unchanged by the squashing.

\section{Summary and outlook}
\label{sec:conclusion}

The main conclusion of this paper is that the continuum limit of the matrix model
built from concentric fuzzy $\CP^3$'s is the metric cone $C(\CP^3)$ over $\CP^3$,
not flat $\mathbb{R}^7$. Since every ingredient was fixed before any comparison was
made, this geometry is a prediction rather than an input. The cone appears with the
innermost layer in place of the apex; it is reached at second order in each fixed
angular sector, it is robust against the placement of the outer boundary, and the
boundary radius is itself selected by the spectrum. At $20480$ layers the sectors
$p\leq8$ agree with the cone to $1\times10^{-9}$, while flat space deviates by
$8$--$18\%$ for $1\leq p\leq8$. The cone and flat $\mathbb{R}^7$ share the $s$-wave
sector, so the geometry becomes visible only in sectors with $p\geq1$, and a test
restricted to $s$ waves cannot tell it from flat space. The identification does not
depend on the scale of the link. The $S^2$ onion gives flat
$\mathbb{R}^3$ because the link of its cone is the round unit sphere; stacking a general
fuzzy space is expected to yield not flat space but the metric cone over that space. The free propagator converges to that of the cone
at the same order sector by sector and, in the symmetric symbol, in position space
(the coherent-state symbol adds a first-order finite-layer correction), and the
propagator from the innermost layer approaches that from the apex. Heat diffusion
expresses the same geometry in thermodynamic terms: the equilibrium plateau
$c_\infty=N_M^2/\sum_jN_j^2$ yields $Mc_\infty\to7$, which reflects the $r^6$ growth
of the measure, and the relaxation times of the angular sectors agree with the
Bessel zeros of the cone to $0.5\%$ already at $M=12$. Sectors with $p$ close to the
cutoff $M-1$, whose chains contain only a few layers, are not resolved.

This conclusion rests on the construction of the interlayer maps. These maps,
which are not evident for $\CP^3$, were obtained in closed form from the polar
decomposition of the second-quantized Schwinger-boson lift, using an operator
identity of~\cite{DHMO}. The coefficient-copying prescription of Kov\'a\v{c}ik and
Tekel, carried out in an explicit tensor basis, reproduces them identically.
Deviations at order $1/j$ shift the centrifugal term, while the completely positive
map $U^{\rm pos}$, which differs at order $j^{-2}$ and preserves coherent-state
symbols, reaches the same continuum limit only at first order. The resulting
radial Laplacian follows from the integration rule with a geometric midpoint and
has a symmetric form with constant off-diagonal entries. For $\CP^1$, away from the
outer boundary, it reduces to the fuzzy onion of~\cite{KT}, and its $U^{\rm pos}$
counterpart to the Laplacian of~\cite{GP}.

The twistor fibration of $\CP^3$ over $S^4$ is invisible in the spectrum of the
onion as constructed: every mode along the fibre is degenerate with an $S^4$
harmonic of the same $p$ and, seen from $S^4$, belongs to a tower of massless
higher-spin fields. Among the $Spin(5)$-invariant deformations squashing the fibre,
those obtainable by rescaling the generators cannot push the fibre-to-base radius
ratio below $1/\sqrt2$, and in that range no five-dimensional energy window opens.
Even when a stronger squashing is imposed from outside, the radial operator of the
surviving sectors remains that of the cone over $\CP^3$.

We close with open problems. The radial kinetic term has been fixed by the
integration rule, and deriving it from a parent matrix model is the first task;
quantum effects of the string modes connecting the layers~\cite{SteinackerString}
and reduction from Moyal space~\cite{GLMV} are possible starting points. The
continuum limit has been established mode by mode, and convergence as operators
remains to be shown. Our tests probe the continuum geometry but not
noncommutativity itself, so the phase structure of the interacting
theory~\eqref{eq:action} is the natural next step, extending the studies of the $S^2$
onion~\cite{KTH,HKT} to $\CP^3$. For the squashing it remains open whether quantum
effects favour squashing the fibre and whether approximate realizations of $h>1$
become exact at large $k$.

\section*{Acknowledgment}
We used an AI assistant (Claude, Anthropic) to improve the numerical programs and for language editing of the manuscript.
The authors are responsible for the interpretation of the results and for the content of the paper.

\let\doi\relax
\bibliographystyle{ptephy}
\bibliography{refs}

@article{BFSS,
    author = "Banks, Tom and Fischler, W. and Shenker, S. H. and Susskind, Leonard",
    title = "{M theory as a matrix model: A conjecture}",
    eprint = "hep-th/9610043",
    archivePrefix = "arXiv",
    reportNumber = "RU-96-95, SU-ITP-96-12, UTTG-13-96",
    doi = "10.1103/PhysRevD.55.5112",
    journal = "Phys. Rev. D",
    volume = "55",
    pages = "5112--5128",
    year = "1997"
}

@article{IKKT,
    author = "Ishibashi, N. and Kawai, H. and Kitazawa, Y. and Tsuchiya, A.",
    title = "{A Large N reduced model as superstring}",
    eprint = "hep-th/9612115",
    archivePrefix = "arXiv",
    reportNumber = "KEK-TH-503",
    doi = "10.1016/S0550-3213(97)00290-3",
    journal = "Nucl. Phys. B",
    volume = "498",
    pages = "467--491",
    year = "1997"
}

@article{Steinacker2010,
    author = "Steinacker, Harold",
    title = "{Emergent Geometry and Gravity from Matrix Models: an Introduction}",
    eprint = "1003.4134",
    archivePrefix = "arXiv",
    primaryClass = "hep-th",
    reportNumber = "UWTHPH-2010-4",
    doi = "10.1088/0264-9381/27/13/133001",
    journal = "Class. Quant. Grav.",
    volume = "27",
    pages = "133001",
    year = "2010"
}

@article{KNT,
    author = "Kim, Sang-Woo and Nishimura, Jun and Tsuchiya, Asato",
    title = "{Expanding (3+1)-dimensional universe from a Lorentzian matrix model for superstring theory in (9+1)-dimensions}",
    eprint = "1108.1540",
    archivePrefix = "arXiv",
    primaryClass = "hep-th",
    reportNumber = "KEK-TH-1484, OU-HET-720-2011",
    doi = "10.1103/PhysRevLett.108.011601",
    journal = "Phys. Rev. Lett.",
    volume = "108",
    pages = "011601",
    year = "2012"
}

@phdthesis{Hoppe,
    author = "Hoppe, Jens",
    title = "{Quantum theory of a massless relativistic surface and a two-dimensional bound state problem}",
    school = "Massachusetts Institute of Technology",
    year = "1982"
}

@article{Madore,
    author = "Madore, J.",
    title = "{The Fuzzy sphere}",
    reportNumber = "LPTHE-ORSAY-91-09",
    doi = "10.1088/0264-9381/9/1/008",
    journal = "Class. Quant. Grav.",
    volume = "9",
    pages = "69--88",
    year = "1992"
}

@article{Myers,
    author = "Myers, Robert C.",
    title = "{Dielectric branes}",
    eprint = "hep-th/9910053",
    archivePrefix = "arXiv",
    reportNumber = "MCGILL-99-27, NSF-ITP-99-113",
    doi = "10.1088/1126-6708/1999/12/022",
    journal = "JHEP",
    volume = "12",
    pages = "022",
    year = "1999"
}

@article{IKTW,
    author = "Iso, Satoshi and Kimura, Yusuke and Tanaka, Kanji and Wakatsuki, Kazunori",
    title = "{Noncommutative gauge theory on fuzzy sphere from matrix model}",
    eprint = "hep-th/0101102",
    archivePrefix = "arXiv",
    reportNumber = "KEK-TH-737, TIT-HEP-462",
    doi = "10.1016/S0550-3213(01)00173-0",
    journal = "Nucl. Phys. B",
    volume = "604",
    pages = "121--147",
    year = "2001"
}

@article{ABNN,
    author = "Azuma, Takehiro and Bal, Subrata and Nagao, Keiichi and Nishimura, Jun",
    title = "{Nonperturbative studies of fuzzy spheres in a matrix model with the Chern-Simons term}",
    eprint = "hep-th/0401038",
    archivePrefix = "arXiv",
    reportNumber = "KUNS-1883, KEK-TH-929",
    doi = "10.1088/1126-6708/2004/05/005",
    journal = "JHEP",
    volume = "05",
    pages = "005",
    year = "2004"
}

@article{Kitazawa,
    author = "Kitazawa, Yoshihisa",
    title = "{Matrix models in homogeneous spaces}",
    eprint = "hep-th/0207115",
    archivePrefix = "arXiv",
    reportNumber = "KEK-TH-833",
    doi = "10.1016/S0550-3213(02)00682-X",
    journal = "Nucl. Phys. B",
    volume = "642",
    pages = "210--226",
    year = "2002"
}

@article{BDLMO,
    author = "Balachandran, A. P. and Dolan, Brian P. and Lee, Joo-Han and Martin, X. and O'Connor, Denjoe",
    title = "{Fuzzy complex projective spaces and their star products}",
    eprint = "hep-th/0107099",
    archivePrefix = "arXiv",
    reportNumber = "SU-4252-740, DIAS-STP-01-11",
    doi = "10.1016/S0393-0440(02)00020-7",
    journal = "J. Geom. Phys.",
    volume = "43",
    pages = "184--204",
    year = "2002"
}

@article{ABIY,
    author = "Alexanian, G. and Balachandran, A. P. and Immirzi, G. and Ydri, B.",
    title = "{Fuzzy $\mathbb{CP}^2$}",
    eprint = "hep-th/0103023",
    archivePrefix = "arXiv",
    reportNumber = "SU-4240-730",
    doi = "10.1016/S0393-0440(01)00070-5",
    journal = "J. Geom. Phys.",
    volume = "42",
    pages = "28--53",
    year = "2002"
}

@article{GR,
    author = "Guralnik, Zachary and Ramgoolam, Sanjaye",
    title = "{On the Polarization of unstable D0-branes into noncommutative odd spheres}",
    eprint = "hep-th/0101001",
    archivePrefix = "arXiv",
    reportNumber = "BROWN-HET-1251, MIT-CTP-3069",
    doi = "10.1088/1126-6708/2001/02/032",
    journal = "JHEP",
    volume = "02",
    pages = "032",
    year = "2001"
}

@article{Ramgoolam,
    author = "Ramgoolam, Sanjaye",
    title = "{On spherical harmonics for fuzzy spheres in diverse dimensions}",
    eprint = "hep-th/0105006",
    archivePrefix = "arXiv",
    reportNumber = "BROWN-HET-1260",
    doi = "10.1016/S0550-3213(01)00315-7",
    journal = "Nucl. Phys. B",
    volume = "610",
    pages = "461--488",
    year = "2001"
}

@article{KTcorfu,
    author = "Kov{\'a}{\v{c}}ik, Samuel and Tekel, Juraj",
    title = "{The Fuzzy Onion: A Proposal}",
    eprint = "2304.14205",
    archivePrefix = "arXiv",
    primaryClass = "hep-th",
    doi = "10.22323/1.436.0312",
    journal = "PoS",
    volume = "CORFU2022",
    pages = "312",
    year = "2023"
}

@article{KT,
    author = "Kov{\'a}{\v{c}}ik, Samuel and Tekel, Juraj",
    title = "{Fuzzy onionlike space as a matrix model}",
    eprint = "2309.00576",
    archivePrefix = "arXiv",
    primaryClass = "hep-th",
    doi = "10.1103/PhysRevD.109.105004",
    journal = "Phys. Rev. D",
    volume = "109",
    number = "10",
    pages = "105004",
    year = "2024"
}

@article{KTH,
    author = "Kov{\'a}{\v{c}}ik, S. and Tekel, J. and Hrmo, M.",
    title = "{The Fuzzy Onion: An Initial Study}",
    eprint = "2405.12393",
    archivePrefix = "arXiv",
    primaryClass = "hep-th",
    doi = "10.22323/1.463.0253",
    journal = "PoS",
    volume = "CORFU2023",
    pages = "253",
    year = "2024"
}

@article{HKT,
    author = "Hrmo, Matej and Kov{\'a}{\v{c}}ik, Samuel and Tekel, Juraj",
    title = "{Numerical Study of Scalar Field Theory on the Fuzzy Onion}",
    eprint = "2608.08855",
    archivePrefix = "arXiv",
    primaryClass = "hep-th",
    month = "8",
    year = "2026"
}

@article{Penrose,
    author = "Penrose, R.",
    title = "{Twistor algebra}",
    doi = "10.1063/1.1705200",
    journal = "J. Math. Phys.",
    volume = "8",
    pages = "345",
    year = "1967"
}

@article{AHS,
    author = "Atiyah, M. F. and Hitchin, Nigel J. and Singer, I. M.",
    title = "{Self-duality in four-dimensional Riemannian geometry}",
    doi = "10.1098/rspa.1978.0143",
    journal = "Proc. Roy. Soc. Lond. A",
    volume = "362",
    pages = "425--461",
    year = "1978"
}

@article{CLT,
    author = "Castelino, Judith and Lee, Sangmin and Taylor, Washington",
    title = "{Longitudinal five-branes as four spheres in matrix theory}",
    eprint = "hep-th/9712105",
    archivePrefix = "arXiv",
    reportNumber = "PUPT-1749",
    doi = "10.1016/S0550-3213(98)00291-0",
    journal = "Nucl. Phys. B",
    volume = "526",
    pages = "334--350",
    year = "1998"
}

@article{MedinaOConnor,
    author = "Medina, Julieta and O'Connor, Denjoe",
    title = "{Scalar field theory on fuzzy $S^4$}",
    eprint = "hep-th/0212170",
    archivePrefix = "arXiv",
    doi = "10.1088/1126-6708/2003/11/051",
    journal = "JHEP",
    volume = "11",
    pages = "051",
    year = "2003"
}

@article{MHOD,
    author = "Medina, Julieta and Huet, Idrish and O'Connor, Denjoe and Dolan, Brian P.",
    title = "{Scalar and Spinor Field Actions on Fuzzy $S^4$: fuzzy $\mathbb{CP}^3$ as a $S^2_F$ bundle over $S^4_F$}",
    eprint = "1208.0348",
    archivePrefix = "arXiv",
    primaryClass = "hep-th",
    doi = "10.1007/JHEP08(2012)070",
    journal = "JHEP",
    volume = "08",
    pages = "070",
    year = "2012"
}

@article{SperlingSteinacker,
    author = "Sperling, Marcus and Steinacker, Harold C.",
    title = "{Covariant 4-dimensional fuzzy spheres, matrix models and higher spin}",
    eprint = "1704.02863",
    archivePrefix = "arXiv",
    primaryClass = "hep-th",
    reportNumber = "UWTHPH-2017-08",
    doi = "10.1088/1751-8121/aa8295",
    journal = "J. Phys. A",
    volume = "50",
    number = "37",
    pages = "375202",
    year = "2017"
}

@book{Steinacker,
    author = "Steinacker, Harold C.",
    title = "{Quantum Geometry, Matrix Theory, and Gravity}",
    doi = "10.1017/9781009440776",
    isbn = "978-1-009-44077-6, 978-1-009-44078-3",
    publisher = "Cambridge University Press",
    year = "2024"
}

@article{ABJM,
    author = "Aharony, Ofer and Bergman, Oren and Jafferis, Daniel Louis and Maldacena, Juan",
    title = "{$\mathcal{N}=6$ superconformal Chern-Simons-matter theories, M2-branes and their gravity duals}",
    eprint = "0806.1218",
    archivePrefix = "arXiv",
    primaryClass = "hep-th",
    reportNumber = "WIS-12-08-JUN-DPP",
    doi = "10.1088/1126-6708/2008/10/091",
    journal = "JHEP",
    volume = "10",
    pages = "091",
    year = "2008"
}

@book{Helgason,
    author = "Helgason, Sigurdur",
    title = "{Groups and Geometric Analysis}",
    publisher = "Academic Press",
    year = "1984"
}

@book{Perelomov,
    author = "Perelomov, Askold",
    title = "{Generalized Coherent States and Their Applications}",
    doi = "10.1007/978-3-642-61629-7",
    isbn = "978-3-540-15912-4, 978-3-642-64891-5, 978-3-642-61629-7",
    publisher = "Springer",
    series = "Theoretical and Mathematical Physics",
    year = "1986"
}

@article{DHMO,
    author = "Dolan, Brian P. and Huet, Idrish and Murray, Sean and O'Connor, Denjoe",
    title = "{Noncommutative vector bundles over fuzzy $\mathbb{CP}^N$ and their covariant derivatives}",
    eprint = "hep-th/0611209",
    archivePrefix = "arXiv",
    doi = "10.1088/1126-6708/2007/07/007",
    journal = "JHEP",
    volume = "07",
    pages = "007",
    year = "2007"
}

@article{GP,
    author = "G{\'a}likov{\'a}, Veronika and Pre{\v{s}}najder, Peter",
    title = "{Coulomb problem in non-commutative quantum mechanics - Exact solution}",
    eprint = "1112.4643",
    archivePrefix = "arXiv",
    primaryClass = "math-ph",
    year = "2011"
}

@article{GKP,
    author = "G{\'a}likov{\'a}, Veronika and Kov{\'a}{\v{c}}ik, Samuel and Pre{\v{s}}najder, Peter",
    title = "{Quantum Mechanics in Noncommutative Space}",
    eprint = "1510.04496",
    archivePrefix = "arXiv",
    primaryClass = "math-ph",
    journal = "Acta Phys. Slov.",
    volume = "65",
    number = "3",
    pages = "153--234",
    year = "2015"
}

@article{HLS,
    author = "Hammou, A. B. and Lagraa, M. and Sheikh-Jabbari, M. M.",
    title = "{Coherent state induced star product on $\mathbb{R}^3_\lambda$ and the fuzzy sphere}",
    eprint = "hep-th/0110291",
    archivePrefix = "arXiv",
    doi = "10.1103/PhysRevD.66.025025",
    journal = "Phys. Rev. D",
    volume = "66",
    pages = "025025",
    year = "2002"
}

@article{VW,
    author = "Vitale, Patrizia and Wallet, Jean-Christophe",
    title = "{Noncommutative field theories on $\mathbb{R}^3_\lambda$: Toward UV/IR mixing freedom}",
    eprint = "1212.5131",
    archivePrefix = "arXiv",
    primaryClass = "hep-th",
    doi = "10.1007/JHEP04(2013)115",
    journal = "JHEP",
    volume = "04",
    pages = "115",
    year = "2013",
    note = "[Addendum: JHEP 03, 115 (2015)]"
}

@article{Hrmo2026,
    author = "Hrmo, M. and Kov{\'a}{\v{c}}ik, S. and Magdolenov{\'a}, K. and Manta, A. and Nede{\v{l}}kov{\'a}, K. and Rusn{\'a}k, P. and Steinacker, H. C. and Tekel, J.",
    title = "{Accretion, mergers, and metastability of fuzzy spheres in a three-matrix model}",
    eprint = "2609.00328",
    archivePrefix = "arXiv",
    primaryClass = "hep-th",
    year = "2026"
}

@article{Sako,
    author = "Sako, Akifumi",
    title = "{Quantization of Algebraic Varieties Defined by Casimir Polynomials via Matrix Regularization: Fuzzy $S^7$ and Beyond}",
    eprint = "2608.04717",
    archivePrefix = "arXiv",
    primaryClass = "hep-th",
    year = "2026"
}

@article{IKK,
    author = "Iso, Satoshi and Kawai, Hikaru and Kitazawa, Yoshihisa",
    title = "{Bilocal fields in noncommutative field theory}",
    eprint = "hep-th/0001027",
    archivePrefix = "arXiv",
    reportNumber = "KEK-TH-669, KUNS-1628",
    doi = "10.1016/S0550-3213(00)00092-4",
    journal = "Nucl. Phys. B",
    volume = "576",
    pages = "375--398",
    year = "2000"
}

@article{SteinackerString,
    author = "Steinacker, Harold C.",
    title = "{String states, loops and effective actions in noncommutative field theory and matrix models}",
    eprint = "1606.00646",
    archivePrefix = "arXiv",
    primaryClass = "hep-th",
    reportNumber = "UWTHPH-2016-9",
    doi = "10.1016/j.nuclphysb.2016.06.029",
    journal = "Nucl. Phys. B",
    volume = "910",
    pages = "346--373",
    year = "2016"
}

@article{SteinackerTekel,
    author = "Steinacker, Harold C. and Tekel, Juraj",
    title = "{String modes, propagators and loops on fuzzy spaces}",
    eprint = "2203.02376",
    archivePrefix = "arXiv",
    primaryClass = "hep-th",
    reportNumber = "UWThPh-2022-3",
    doi = "10.1007/JHEP06(2022)136",
    journal = "JHEP",
    volume = "06",
    pages = "136",
    year = "2022"
}

@article{GLMV,
    author = "Gracia-Bondia, Jose M. and Lizzi, F. and Marmo, G. and Vitale, P.",
    title = "{Infinitely many star products to play with}",
    eprint = "hep-th/0112092",
    archivePrefix = "arXiv",
    reportNumber = "DSF-40-01",
    doi = "10.1088/1126-6708/2002/04/026",
    journal = "JHEP",
    volume = "04",
    pages = "026",
    year = "2002"
}

@article{DSS,
    author = "DeBellis, Joshua and Saemann, Christian and Szabo, Richard J.",
    title = "{Quantized Nambu-Poisson Manifolds and n-Lie Algebras}",
    eprint = "1001.3275",
    archivePrefix = "arXiv",
    primaryClass = "hep-th",
    reportNumber = "HWM-10-3, EMPG-10-03",
    doi = "10.1063/1.3503773",
    journal = "J. Math. Phys.",
    volume = "51",
    pages = "122303",
    year = "2010"
}

@article{KDR,
    author = "Kov{\'a}{\v{c}}ik, Samuel and {\v{D}}ur{\'\i}{\v{s}}kov{\'a}, Michaela and Rusn{\'a}k, Patrik",
    title = "{Phenomenology of the dispersion law in a three-dimensional quantum space}",
    eprint = "2402.05832",
    archivePrefix = "arXiv",
    primaryClass = "gr-qc",
    doi = "10.1142/S0217751X24501549",
    journal = "Int. J. Mod. Phys. A",
    volume = "40",
    number = "01",
    pages = "2450154",
    year = "2025"
}

@article{AV,
    author = "Acharyya, Nirmalendu and Vaidya, Sachindeo",
    title = "{Fuzzy Conifold $Y_F^6$ and Monopoles on $S_F^2\times S_F^2$}",
    eprint = "1307.6079",
    archivePrefix = "arXiv",
    primaryClass = "hep-th",
    doi = "10.1142/S0217751X14500067",
    journal = "Int. J. Mod. Phys. A",
    volume = "29",
    pages = "1450006",
    year = "2014"
}

@book{BZ,
    author = "Bengtsson, Ingemar and {\.Z}yczkowski, Karol",
    title = "{Geometry of Quantum States: An Introduction to Quantum Entanglement}",
    edition = "2nd",
    publisher = "Cambridge University Press",
    year = "2017"
}

@article{Cheeger,
    author = "Cheeger, Jeff",
    title = "{Spectral geometry of singular Riemannian spaces}",
    doi = "10.4310/jdg/1214438175",
    journal = "J. Diff. Geom.",
    volume = "18",
    number = "4",
    pages = "575--657",
    year = "1983"
}

\appendix

\section{Deformations of the Up map and the effective centrifugal term}
\label{app:centrifugal}

The effect of a normalization $t_p(j)=\sqrt{N_{j+1}/N_j}\,(1+\varepsilon_{p,j})$ that
violates (ii) is seen by writing the coefficients of sector $p$ as
$c_j=\gamma_j\beta_j$ with $\gamma_{j+1}/\gamma_j=1+\varepsilon_{p,j}$: the deformed chain is
equivalent to the undeformed one with the measure $N_j^2$ multiplied by $\gamma_j^2$
and the bond weight $N_jN_{j+1}$ by $\gamma_{j+1}^2$. For $\varepsilon_{p,j}=\delta_p/j$ one has
$\gamma_j\propto r^{\delta_p}$, so the $\beta$ chain tends to
$\partial_r^2+\frac{6+2\delta_p}{r}\partial_r$; returning to $c=r^{\delta_p}\beta$, the
dimension stays seven and the effective angular eigenvalue shifts to
\begin{equation}\label{eq:Ceff}
C_2^{\rm eff}(p)=p(p+3)+\delta_p(\delta_p+5).
\end{equation}
For $\varepsilon_{p,j}=O(j^{-2})$, $\gamma_j$ converges and the limit is unchanged, but
convergence generally drops to first order ($U^{\rm pos}$ is an example).

Consistency with the angular Laplacian~\eqref{eq:KL}, which acts within each layer,
requires $\delta_p(\delta_p+5)=0$. The root $\delta_p=0$ is (ii); the root
$\delta_p=-5$ gives the same operator in the bulk but changes the Neumann condition
$xJ_\nu'(x)=\frac52J_\nu(x)$ at the outer end into the Robin condition
$xJ_\nu'(x)=-\frac52J_\nu(x)$. This residual freedom is removed by the
coefficient-copying prescription of Theorem~\ref{thm:copymain}.

\section{The linear program for the squashing}
\label{app:lp}

This appendix gives the procedure behind the statements on arbitrary configurations
in Section~\ref{sec:kk-setup}. Fix the level $k=j-1$. Under $Sp(2)$, $\End(V_k)$
decomposes into the components $V_{(p,s)}\cong(p-s,2s)$, $0\leq s\leq p\leq k$, each
occurring once (Section~\ref{sec:kk-setup}). The target $\hat C_2+h\Delta_I$ acts on
$V_{(p,s)}$ as $\tau_{(p,s)}(h)=p(p+3)+h\,s(s+1)$.

Let $\{X_a\}$ be any configuration of Hermitian matrices on $V_k$ and
$\Delta_X=\sum_a[X_a,[X_a,\cdot\,]]$ its Laplacian. Since the target is
$Sp(2)$-invariant, $\Delta_X$ may be replaced by its $Sp(2)$ average, which is the
Laplacian of the configuration together with all its $Sp(2)$ rotations. The average
acts on $V_{(p,s)}$ as a scalar $\lambda_{(p,s)}$. Write $X_a=\sum_{(q,t)}X_a^{(q,t)}$
with $X_a^{(q,t)}\in V_{(q,t)}$. Under the average, cross terms between different
components vanish, and by Schur's lemma each component contributes in proportion to
its norm, so that
\begin{equation}\label{eq:lp}
\lambda_{(p,s)}=\sum_{(q,t)}L_{(p,s),(q,t)}\,g_{(q,t)},\qquad
L_{(p,s),(q,t)}=\sum_{B}\bigl\|[B,O_{(p,s)}]\bigr\|^2_{\HS},\qquad
g_{(q,t)}=\frac{\sum_a\|X_a^{(q,t)}\|^2_{\HS}}{\dim V_{(q,t)}}\geq0 .
\end{equation}
Here $B$ runs over a HS-orthonormal basis of $V_{(q,t)}$ and $O_{(p,s)}$ is any unit
element of $V_{(p,s)}$. Conversely, every $g\geq0$ is realized by the configuration
$\{\sqrt{g_{(q,t)}}\,B\}$, with a Hermitian basis $B$ of each component. The components
with $p=0$ or $q=0$ are multiples of the identity and drop out, so $L$ is a square matrix
of size $k(k+3)/2$.

Exact realization of the squashing is therefore the feasibility problem $Lg=\tau(h)$,
$g\geq0$. We computed $L$ for $k=1,\dots,6$ by decomposing $\End(V_k)$ numerically
(the components are the eigenspaces of $\hat C_2+\eta\,\hat C_2^{Spin(5)}$ for a
generic $\eta$). In all cases $L$ is non-singular, with condition numbers $4$,
$26$, $89$, $235$, $526$ and $1051$. The unique solution is
\begin{equation}\label{eq:lpsol}
g_{(1,0)}=I_2(k)\,(1-h),\qquad g_{(1,1)}=I_2(k)\,(1+h),\qquad g_{(q,t)}=0\ \ (q\geq2),
\end{equation}
with $I_2(k)=\frac12\binom{k+4}{5}$ the norm of the generators (footnote in
section~\ref{sec:boson}). This is the rescaled configuration of
section~\ref{sec:kk-setup}, and it is non-negative exactly for $-1\leq h\leq1$. For
$k=2$, with rows and columns ordered as $(1,0)$, $(1,1)$, $(2,0)$, $(2,1)$, $(2,2)$,
\begin{equation}\label{eq:Lk2}
L=\frac13\begin{pmatrix}
2&2&14&28&14\\
1&3&7&21&28\\
5&5&5&25&20\\
4&6&10&18&22\\
2&8&8&22&20
\end{pmatrix},
\end{equation}
and~\eqref{eq:lpsol} with $I_2(2)=3$ can be checked by hand. The ratio of the rows
$(2,2)$ and $(2,0)$ lies between $0.4$ and $1.6$ in every column, which gives the bound
$\lambda_{(2,2)}\leq1.6\,\lambda_{(2,0)}$ of Figure~\ref{fig:squash}(b).

For the approximate realization we fix $h$ and minimize $\varepsilon$ subject to
$g\geq0$ and $|(Lg)_{(p,s)}-\tau_{(p,s)}(h)|\leq\varepsilon\,\tau_{(p,s)}(h)$ for all rows
with $p\leq2$. This linear program gives the values of Figure~\ref{fig:squash}(c).
\end{document}